\documentclass[11pt]{article}

\usepackage[utf8]{inputenc}
\usepackage[english]{babel}

\usepackage{amsmath,amsthm,amssymb,mathtools,mathrsfs,bbm}

\DeclareMathOperator{\Var}{Var}
\DeclareMathOperator{\vech}{vech} 

\usepackage{booktabs,threeparttable,siunitx,multirow}
\usepackage[font=small,labelfont=bf]{caption}
\usepackage{subcaption}
\usepackage{float} 
\usepackage{tikz}
\usetikzlibrary{decorations.pathreplacing,arrows,positioning}

\usepackage[ruled,vlined]{algorithm2e}

\usepackage{commath,changepage,url,xfrac,enumerate,authblk,charter,xr}

\usepackage[a4paper,margin=25mm]{geometry}
\usepackage{setspace}
\usepackage[parfill]{parskip}

\usepackage[authoryear,round]{natbib} 

\newtheorem{theorem}{Theorem}[section]
\newtheorem{assumption}{Assumption}[section]
\newtheorem{corollary}{Corollary}[section]
\newtheorem{lemma}{Lemma}[section]

\theoremstyle{definition}
\newtheorem{definition}{Definition}[section]
\newtheorem{remark}{Remark}[section]

\providecommand{\keywords}[1]{\small\textbf{\textit{Keywords:}} #1}
\newcommand{\jel}[1]{\small\textbf{\textit{JEL:}} #1}

\usepackage[hidelinks]{hyperref}

\title{Identification and Estimation of Simultaneous Equation Models Using Higher-Order Cumulant Restrictions}
\author{Ziyu Jiang \\ UCL}
\date{\today}

\begin{document}

\maketitle

\begin{abstract}
Identifying structural parameters in linear simultaneous-equation models is a longstanding challenge. Recent work exploits information in higher-order moments of non-Gaussian data. In this literature, the structural errors are typically assumed to be uncorrelated so that, after standardizing the covariance matrix of the observables (whitening), the structural parameter matrix becomes orthogonal—a device that underpins many identification proofs but can be restrictive in econometric applications. We show that neither zero covariance nor whitening is necessary. For any order \(h>2\), a simple diagonality condition on the \(h\)th-order cumulants \emph{alone} identifies the structural parameter matrix—up to unknown scaling and permutation—as the solution to an eigenvector problem; no restrictions on cumulants of other orders are required. This general, single-order result enlarges the class of models covered by our framework and yields a sample-analogue estimator that is \(\sqrt{n}\)-consistent, asymptotically normal, and easy to compute. Furthermore, when uncorrelatedness is intrinsic—as in vector autoregressive (VAR) models—our framework provides a transparent overidentification test. Monte Carlo experiments show favorable finite-sample performance, and two applications—\emph{Returns to Schooling} and \emph{Uncertainty and the Business Cycle}—demonstrate its practical value.

\end{abstract}


\jel{C10, C30, C32, C38}

\section{Introduction and Literature Review}

Linear simultaneous-equation models are among the most commonly used tools in economics. Their appeal lies in their ability to capture equilibrium relationships and other scenarios where variables are determined jointly. However, as illustrated by the canonical ``supply and demand'' system, the structural equations that encode these relationships are typically not identified.

As \citet{HAUSMAN1983} discusses, the predominant solution to this identification challenge has been the instrumental variables (IV) approach. Yet, finding valid instruments can be difficult, as it often requires extensive structural modeling, deep subject-matter knowledge, and considerable ingenuity. This difficulty motivates a natural question: Is there a pathway forward when instruments are unavailable?

Such concerns have appeared frequently in econometric applications. A classic non-identification result arises in the vector autoregression (VAR) literature: if no further restrictions are imposed and one assumes a simple setting in which structural errors are uncorrelated and normally distributed, then the distribution of reduced-form errors---completely characterized by their mean and covariance---is insufficient to identify the structural parameters (see \cite{LANNE2017288} for a detailed argument). A common remedy involves exploiting higher-order information, such as identification from changes in second moments (heteroskedasticity), which has spawned a rich body of literature proposing influential identification strategies (e.g., \citet{10.1162/003465303772815727,LANNE2010121}; see \citet{annurev:/content/journals/10.1146/annurev-economics-070124-051419} for a comprehensive review). Here, we pursue higher-order moments in a more direct manner---namely, by moving away from the Gaussian assumption and focusing on distributions that encode additional information in their higher moments, such as skewness and kurtosis. Pioneering work in this direction can be traced back to \citet{bbfc08fa-3d75-37e2-b3fe-02d4101a7538} and \citet{22ff115f-e119-3e81-8033-1597bfc7003b}, who showed that, in errors-in-variables settings, non-Gaussianity can deliver identification without instruments. 

Beyond econometrics, the assumption of mutually independent, non-Gaussian errors has also been widely investigated in signal processing. Under this framework, \citet{COMON1994287} developed one of the most influential blind source separation methods, later termed independent component analysis (ICA). In the simultaneous-equations context, ICA theory implies that if the structural errors are mutually independent and non-Gaussian (with at most one Gaussian component), then the structural parameter matrix (which defines the linear relationships across equations) is identified up to signed permutation and scale. Although this does not provide full point identification in every case, it is still a powerful starting point that, coupled with basic and relatively uncontroversial economic theory, can deliver exact identification.

In principle, ICA identification results translate into viable estimators. One prominent example is the FastICA algorithm introduced by \citet{761722}, which has been widely applied in fields ranging from telecommunications to medical imaging \citep{hyvarinen2013independent}. Nevertheless, within econometrics---especially in microeconometrics---ICA-based estimators have seen limited use. One reason is that the statistical properties of these estimators are often not the central focus of ICA research, with \citet{10.1214/009053606000000939} standing out as a notable exception. A more fundamental modeling concern arises from how structural errors are typically viewed in econometrics: rather than ``signals'' broadcasting independently, they are usually conceived as latent factors, which makes the strong assumption of mutual independence seem less realistic.

\citet{bonhomme2009consistent} are among the earliest econometric attempts to move beyond the strict independence setup in a linear factor model with additive error, allowing error terms that may be correlated to some degree. Working in this setting, their identification argument shows that second- to fourth-order information can identify part or all of the loading/structural matrix, with the identifiable portion shrinking as error correlation strengthens; they also propose a JADE-type joint-diagonalization estimator that is consistent and asymptotically normal. Subsequently, in the SVAR literature, \citet{lanne2021gmm} and \citet{guay2021identification} show that, while maintaining uncorrelated structural errors, identification follows from imposing diagonality of either the third- or fourth-order cumulant tensor.\footnote{Formal definitions appear in Section~2.} \citet{mesters2024non} extend this line of work, still assuming uncorrelated structural errors, by allowing richer dependence structures beyond diagonal fourth cumulants. 

A recurring theme in this literature is the assumption that structural errors are uncorrelated. The chief reason is that it enables whitening, a preprocessing step that transforms the variables so that the mixing/structural‑coefficient matrix becomes orthogonal. This whitening step underpins most existing identification proofs. In structural VAR (SVAR) settings the assumption is usually innocuous, because uncorrelated shocks are both standard and, to some extent, desirable. In many other simultaneous-equation settings, however, the assumption is overly restrictive: it rules out linear dependence among structural errors, including dependence induced by measurement error or omitted common shocks. For instance, if an omitted causal variable appears in two equations, it induces correlation between their structural errors, violating the assumption. 

This paper's first contribution is to show that the zero-covariance assumption is unnecessary for structural-parameter identification when a diagonal condition on a single higher-order moment (cumulant) is available. In Section~5 we prove that, under a diagonal higher-order cumulant assumption, whitening is unnecessary: the identification problem reduces to a simple eigenvector problem. This argument is both simple and constructive: it leads directly to a straightforward sample-analogue estimator. Moreover, even in settings such as SVAR models—where relaxing the uncorrelatedness assumption is arguably less crucial—our framework still offers a transparent test of that uncorrelatedness within the same higher-order cumulant framework used in earlier work.

Our second contribution is to propose a conceptually and computationally simple estimator with desirable statistical properties. In Section~6, we establish that the sample analogue estimator derived from our identification argument is consistent and asymptotically normal (under standard regularity conditions). We also show how the identification framework simplifies both the asymptotic analysis and the implementation. Section~7 then evaluates the estimator’s finite-sample performance in Monte Carlo experiments, and we conclude by illustrating its practical usefulness through two empirical applications.

\section{Preliminaries: Higher-Order Cumulants}

Let $X \in \mathbb{R}^{d}$ be a random vector. Its characteristic function is
\begin{equation}\label{eq:charfun}
\phi_{X}(t)\;=\; \mathbb{E}\!\left[\exp\bigl(\mathrm{i}\,t^{\mathsf T}X\bigr)\right],
\qquad t \in \mathbb{R}^{d}.
\end{equation}

\noindent
The cumulant generating function is
\begin{equation}\label{eq:CGF}
\psi_{X}(t)
  \;=\; \log \phi_{X}(t)
  \;=\; \sum_{r=1}^{\infty} \frac{\mathrm{i}^{\,r}}{r!}
        \sum_{i_{1},\dots,i_{r}=1}^{d}
        \!\!\bigl(t_{i_{1}}\cdots t_{i_{r}}\bigr)\,
        C_{r}\!\bigl(X_{i_{1}},\dots,X_{i_{r}}\bigr).
\end{equation}
Here each index $i_{1},\dots,i_{r}$ runs independently over $\{1,\dots,d\}$, so the inner sum contains $d^{r}$ terms.

\begin{equation}\label{eq:cumulant-def}
C_{r}(X_{i_{1}},\dots,X_{i_{r}})
    \;=\;\frac{1}{\mathrm{i}^{\,r}}\,
          \biggl.\frac{\partial^{r} \psi_{X}(t)}
                   {\partial t_{i_{1}}\cdots\partial t_{i_{r}}}\biggr|_{t=0},
\end{equation}
which defines the \emph{$r$th-order cumulants}. Collecting all such coefficients yields the $r$th-order tensor $C_{r}(X)\in(\mathbb{R}^{d})^{\otimes r}$.

For expositional clarity, the main text focuses on the third- and fourth order cases; extensions to any fixed order $r\ge 3$ are stated and proved in the appendix. Write $\mu_i:=\mathbb{E}[X_i]$ and $\mathrm{Cov}(X_a,X_b):=\mathbb{E}\!\big[(X_a-\mu_a)(X_b-\mu_b)\big]$. Then, assuming the relevant moments exist,
\begin{equation}\label{eq:C3-nd}
\begin{aligned}
C_{3}(X)_{ijk}
&= \mathbb{E}\!\big[(X_i-\mu_i)(X_j-\mu_j)(X_k-\mu_k)\big]  \\
&= \mathbb{E}[X_i X_j X_k]
   - \mu_i\,\mathbb{E}[X_j X_k]
   - \mu_j\,\mathbb{E}[X_i X_k]
   - \mu_k\,\mathbb{E}[X_i X_j]
   + 2\,\mu_i\mu_j\mu_k \, ,
\end{aligned}
\end{equation}
\begin{equation}\label{eq:C4-nd}
\begin{aligned}
C_{4}(X)_{ijkl}
&= \mathbb{E}\!\big[(X_i-\mu_i)(X_j-\mu_j)(X_k-\mu_k)(X_l-\mu_l)\big] \\
&\quad - \mathrm{Cov}(X_i,X_j)\,\mathrm{Cov}(X_k,X_l)
      - \mathrm{Cov}(X_i,X_k)\,\mathrm{Cov}(X_j,X_l)
      - \mathrm{Cov}(X_i,X_l)\,\mathrm{Cov}(X_j,X_k) \, .
\end{aligned}
\end{equation}
In the special case $\mathbb{E}[X]=0$, these reduce to
\begin{equation}\label{eq:C3}
C_{3}(X)_{ijk}
    \;=\; C_{3}(X_{i},X_{j},X_{k})
    \;=\; \mathbb{E}\!\bigl[X_{i}X_{j}X_{k}\bigr],
\end{equation}
\begin{equation}\label{eq:C4}
\begin{aligned}
C_{4}(X)_{ijkl}
  \;=\; \mathbb{E}\!\bigl[X_{i}X_{j}X_{k}X_{l}\bigr]
          -\mathbb{E}[X_{i}X_{j}]\,\mathbb{E}[X_{k}X_{l}]
          -\mathbb{E}[X_{i}X_{k}]\,\mathbb{E}[X_{j}X_{l}]
          -\mathbb{E}[X_{i}X_{l}]\,\mathbb{E}[X_{j}X_{k}] \, .
\end{aligned}
\end{equation}

By convention, the entries of the third cumulant tensor $C_3(X)_{ijk}$ with $i=j=k$ are referred to as diagonal entries. Denoting the $i$th diagonal entry by $\kappa_3(X_i)$, we have
\[
\kappa_3(X_i)=C_3(X)_{iii}=\mathbb{E}\!\big[(X_i-\mu_i)^3\big],
\]
the skewness of $X_i$. Similarly, the $i$th diagonal entry of the fourth cumulant tensor is
\[
\kappa_4(X_i)=C_4(X)_{iiii}=\mathbb{E}\!\big[(X_i-\mu_i)^4\big]-3\Big(\mathbb{E}\!\big[(X_i-\mu_i)^2\big]\Big)^2,
\]
commonly called the excess kurtosis.\footnote{%
  In this paper, ``skewness'' and ``excess kurtosis'' denote the \emph{raw} third and
  fourth cumulants; the variance-normalized versions are not used.}

Cumulants have many useful properties. The following list records those used in this paper.

\begin{enumerate}
    \item \textbf{Symmetry.} Cumulant tensors are symmetric: they are invariant under permutation of indices. For example, $C_3(X)_{ijk}=C_3(X)_{ikj}=C_3(X)_{jki}=\cdots$.

    \item \textbf{Translation invariance.} For any deterministic $m\in\mathbb{R}^d$ and any $r\ge 2$,
    \[
      C_r(X+m)=C_r(X).
    \]
    In particular, the choice to center or not only affects how formulas are written; cumulants of order $r\ge 2$ are unchanged by shifts.

    \item \textbf{Multilinearity.} For real constants $\alpha,\beta$ and random variables $U_1,\dots,U_r,V,W$,
    \[
      C_r(U_1,\dots,\alpha V+\beta W,\dots,U_r)
      \;=\;\alpha\,C_r(U_1,\dots,V,\dots,U_r)+\beta\,C_r(U_1,\dots,W,\dots,U_r).
    \]
    Homogeneity of diagonal elements follows:
    \begin{equation}
    \begin{aligned}
       \kappa_3(\alpha X_i)=\alpha^3\kappa_3(X_i) \\
        \kappa_4(\alpha X_i)=\alpha^4\kappa_4(X_i) 
        \label{eq:Homogeneity}
    \end{aligned}
    \end{equation}

     If all cross-cumulants of $(X_i,X_j)$ are zero, then separability holds:
    \begin{equation}
       \kappa_r(X_i+X_j)=\kappa_r(X_i)+\kappa_r(X_j).
       \label{eq:separability}
    \end{equation}

    \item \textbf{Independence and mean independence.}
    If the components of $X$ are mutually independent, then all off-diagonal entries of all cumulant tensors vanish. For $r=3,4$ this is immediate from \eqref{eq:C3-nd}–\eqref{eq:C4-nd}; for general $r$ see \citet{012b02ed-d7ee-3600-b1f1-5d16bbb9ba81}.  
    Moreover, if for each $i$ we have \emph{mean independence}, $\mathbb{E}[X_i\mid X_{-i}]=\mathbb{E}[X_i]$ (where $X_{-i}=(X_j)_{j\ne i}$), then all off-diagonal entries of $C_3(X)$ are zero. For example, for pairwise distinct $i,j,k$,
    \[
      \mathbb{E}\!\big[(X_i-\mu_i)(X_j-\mu_j)(X_k-\mu_k)\big]
      = \mathbb{E}\!\big[(X_j-\mu_j)(X_k-\mu_k)\,\mathbb{E}(X_i-\mu_i\mid X_{-i})\big]
      = 0.
    \]

    \item \textbf{Gaussian and elliptical families.}
    For Gaussian random vectors, $C_r(X)=0$ for all $r\ge3$. For elliptical random vectors, all odd-order cumulants vanish for $r\ge3$ (whenever the corresponding moments exist); see \citet{012b02ed-d7ee-3600-b1f1-5d16bbb9ba81} and \citet{rao2021asymptotic}.
\end{enumerate}



\section{Baseline Model Setup}
This section introduces the setup for our baseline model. Although the assumptions used here are already weaker than those in much of the current literature, they can be relaxed further. For clarity, we postpone the discussion of possible relaxations until Section~5.

Let $S$ be an unobserved $d$-dimensional random vector of structural errors, where $S_q$ denotes its $q$-th component. We observe a $d$-dimensional random vector $X$ generated by the linear structural equation system
\[
\Lambda X = S.
\]

Our primary goal is to recover the structural parameter matrix $\Lambda$ using $n$ independent and identically distributed (i.i.d.) observations of $X$.

The model imposes the following restrictions on $\Lambda$ and a fixed $h$th \((h > 2)\) cumulant tensor of $S$:
\begin{itemize}
    \item[A1)] \textbf{Invertibility:} The $d \times d$ matrix $\Lambda$ is nonsingular, with its inverse denoted $\Lambda^{-1} = A$.
    \item[A2)] \textbf{Nonzero diagonal entries:} $\kappa_h(S_{q}) \neq 0$ for all $q \in \{1,\dots,d\}$.
    \item[A3)] \textbf{Zero off-diagonal entries:} The off-diagonal entries of $C_h(S)$ are zero; i.e., $C_h(S)_{i_1,\dots,i_h}=0$ whenever not all indices $i_1,\dots,i_h$ are equal.
\end{itemize}

Before delving into these assumptions, two observations are in order. First, write $\mu_X:=\mathbb{E}[X]$ and $\mu_S:=\mathbb{E}[S]$, so $\Lambda \mu_X=\mu_S$, and define the centered variables
\[
X_c := X-\mu_X,\qquad S_c := S-\mu_S,\qquad \text{so that}\quad \Lambda X_c = S_c.
\]
In the upcoming sections, for estimation, we preprocess by subtracting the sample mean, and in the population, we work with $X_c$ and $S_c$. Because cumulants of order $h\ge 2$ are translation invariant, we have
\[
C_h(X)=C_h(X_c),\qquad C_h(S)=C_h(S_c),\qquad \kappa_h(S)=\kappa_h(S_c),
\]
so working with the mean zero variables $X_c$ and $S_c$ yields simpler and clearer arguments without loss of generality for both identification and estimation. Second, although the current model representation lacks exogenous observable covariates, these can be incorporated through a simple partialling-out argument.

Assumption A1) is standard: the observable variables $X$ can be written as a linear combination of the structural errors $S$ via the matrix inverse $\Lambda^{-1} = A$. The matrix $A$ is commonly referred to as the mixing matrix. While this assumption is natural in simultaneous-equation models, it is less convenient in factor-model settings where the loading/mixing matrix is typically rectangular. In Section~5, we show how to relax A1 to require only that $A$ have full column rank.

Assumption A3) implies that the $h$th cumulant tensor $C_h(S)$ is diagonal; A2) additionally rules out zero diagonal entries. Note that the commonly used non‑Gaussian, mutually independent assumption for structural errors is stronger than what is imposed here: if all components of $S$ are mutually independent, then (by the standard additivity property of cumulants) all order-$r\ge 3$ cumulant tensors of $S$ are diagonal. Furthermore, if these components deviate from Gaussianity due to asymmetry (nonzero skewness) or heavy or light tails (nonzero excess kurtosis), the third or fourth cumulant will have nonzero diagonal entries. Many classical ICA algorithms and identification protocols based on higher-order information, either implicitly or explicitly, exploit these diagonal structures. In particular, the diagonality of the fourth cumulant, alongside zero covariance (i.e., a diagonal second cumulant), is often the central assumption guaranteeing identification.

As is evident from our assumptions, we depart from the conventional approach—which typically combines a diagonal covariance (second cumulant) with a diagonal order-$h$ cumulant—and show that diagonality of a single higher-order cumulant by itself is sufficient for both identification and estimation, without imposing any second-moment restriction. This relaxation is important not only because it grants greater modeling flexibility, but also because it implies that previously used specifications are overidentified and hence testable. For concreteness, we focus on the case $h=3$, with occasional references to the kurtosis case ($h=4$). There are two main reasons for this choice. First, the structure of theorems and proofs remains the same for any fixed $h>2$, but skewness ($h=3$) yields the cleanest and most interpretable representation. Second, the $h=3$ case benefits most from relaxing full independence, thereby enabling model setups that were previously unattainable.

Under $h=3$, the model assumptions specialize to:
\begin{itemize}
    \item[A1$^\prime$)] \textbf{Invertibility:} The $d \times d$ matrix $\Lambda$ is nonsingular, with its inverse denoted $\Lambda^{-1} = A$.
    \item[A2$^\prime$)] \textbf{Nonzero diagonal entries:} $\kappa_3(S_q) \neq 0$ for all $q$.
    \item[A3$^\prime$)] \textbf{Zero off-diagonal entries:} $C_{3}(S_i, S_j, S_k) 
= \mathbb{E}\!\left[(S_i - \mathbb{E}[S_i])(S_j - \mathbb{E}[S_j])(S_k - \mathbb{E}[S_k])\right] 
= 0 $ unless  i = j = k.
\end{itemize}

Below, we list a few examples for which these assumptions hold, alongside their corresponding econometric applications.

\begin{itemize}
\item \textbf{Independent structural errors and independent factor models.}

In SVAR models that rely on higher-order moment restrictions for identification, it is common to assume that the structural errors are independent\footnote{For third-order identification this can be relaxed to (conditional) mean independence; see Section~2.} (see, for example, \citet{LANNE2017288}, \citet{gourieroux2017statistical}, and \citet{DAVIS2023180}). Although later work has shown that identification can be achieved under weaker conditions—specifically, Assumptions A2$^\prime$) and A3$^\prime$) with $h=3$ or $h=4$ combined with uncorrelated structural errors—those requirements remain more restrictive than our framework. Consequently, the identification and estimation methods developed in this paper also apply in these setups.

It should be noted that in this literature, relaxing the uncorrelatedness assumption is often neither meaningful nor desirable, since doing so complicates the interpretation of impulse response functions. Nonetheless, a principal advantage of using Assumptions A1$^\prime$)–A3$^\prime$) is that they provide a direct way to test the more restrictive model specifications commonly employed in the literature.

Another important class in this category is the family of independent factor models. An example for practical application is \citet{CHEN2010255}, where independent factors are extracted using ICA to facilitate risk analysis of financial portfolios. Although most applications of independent factor models do not satisfy A1$^\prime$) (because the loading/mixing matrix is typically rectangular), the extension discussed in Section~5 addresses this limitation by allowing rectangular mixing matrices of full column rank.

\item \textbf{Composite structural error: skewed equation shifters + symmetric independent common effects}

Let $S$ represent a composite error,
\[
    S = \tilde S + \eta.
\]
Let $\tilde S$ have components such that $C_3(\tilde S)$ is diagonal with nonzero diagonal entries (e.g., mutually independent and skewed).
Assume that $\eta$ is generated by an independent factor model:
\begin{equation}
    \eta = \Gamma e,
\end{equation}
where the components of $e$ are mutually independent and \emph{centrally symmetric} (so $\kappa_3(e_\ell)=0$ for all $\ell$), and $e$ is independent of $\tilde S$.
From a modeling perspective, this specification imposes that the unmodeled (unobserved) common effect is comprised of a linear combination of a series of independent, symmetric effects, while each equation also has idiosyncratic skewed shifters. This data-generating process permits both correlation and common components across equations.

To see how this satisfies Assumptions A2$^\prime$–A3$^\prime$, consider the $(i,j,k)$th entry of the third \emph{cumulant} tensor:
\begin{equation}
\begin{aligned}
    C_3(S_i,S_j,S_k)
    &= C_3\!\Bigl(\tilde S_i + \sum_\ell \Gamma_{i \ell} e_\ell,\;
                   \tilde S_j + \sum_\ell \Gamma_{j \ell} e_\ell,\;
                   \tilde S_k + \sum_\ell \Gamma_{k \ell} e_\ell\Bigr) \\
    &= C_3(\tilde S_i,\tilde S_j,\tilde S_k)
       + \sum_\ell \Gamma_{i \ell}\Gamma_{j \ell}\Gamma_{k \ell}\,\kappa_3(e_\ell),
\end{aligned}
\end{equation}
where independence between $\tilde S$ and $e$, and among the $e_\ell$, eliminates mixed cross‑cumulants. Under the stated symmetry of $e_\ell$, $\kappa_3(e_\ell)=0$; hence the second term vanishes. The first term is nonzero if and only if $i=j=k$, and its diagonal entries are nonzero by construction of $\tilde S$.

An example of an econometric model in this category is the classical measurement‑error model with a skewed true variable. 
Let $\tilde x = x + e$ with classical measurement error $e$ that is symmetric and independent of $(x,\epsilon)$. We therefore have the following equation system:
\begin{equation}
\begin{aligned}
    \tilde{x} &= x + e, \\
    y &= \beta x + \epsilon \;=\; \beta \tilde{x} - \beta e + \epsilon.
\end{aligned}
\end{equation}
To see how this is a special case of the composite error setup, define $S_1 := x + e$ and $S_2 := \epsilon - \beta e$, so $S = (\tilde S_1,\tilde S_2) + (\eta_1,\eta_2)$ with $\tilde S=(x,\epsilon)$ (skewed idiosyncratic shifters) and $\eta=(e,-\beta e)$ (symmetric common effects).
To align with the model $\Lambda X = S$, set
\[
\Lambda = \begin{pmatrix} 1 & 0 \\ -\beta & 1 \end{pmatrix},
\qquad
X = \begin{pmatrix} \tilde x \\ y \end{pmatrix},
\qquad
S = \begin{pmatrix} x + e \\ \epsilon - \beta e \end{pmatrix}.
\]
The above analysis shows that the standard single‑regressor measurement‑error model is a special case of our composite‑error specification. Under the assumptions—$e$ is symmetric and independent of $(x,\epsilon)$, and $x$ and $\epsilon$ are independent—the model satisfies A3$^\prime$. It satisfies A2$^\prime$ when both $x$ and $\epsilon$ are skewed.

\emph{Importantly, even in this simplest case the traditional higher‑moment/ICA strategy fails.}
With $S_1:=x+e$ and $S_2:=\epsilon-\beta e$,
\[
\operatorname{Cov}(S)
=
\begin{pmatrix}
\operatorname{Var}(x)+\operatorname{Var}(e) & -\,\beta\,\operatorname{Var}(e)\\[2pt]
-\,\beta\,\operatorname{Var}(e) & \operatorname{Var}(\epsilon)+\beta^2\operatorname{Var}(e)
\end{pmatrix},
\]
so $\operatorname{Cov}(S)$ is non‑diagonal whenever $\beta\neq 0$ and $\operatorname{Var}(e)>0$.
Therefore, strategies that depend on the uncorrelatedness or mutual independence of the structural errors (e.g., whitening followed by joint diagonalization) are generally invalid in the presence of measurement error.
In contrast, because $e$ is centrally symmetric,
\[
C_3(S)=C_3(\tilde S),
\]
so the third‑order cumulant remains diagonal and our identification based on a single higher‑order cumulant remains valid.

In many applications, researchers face more complex measurement‑error structures—for example, multiple endogenous variables measured with error where the recording processes are intrinsically linked (see, e.g., \citet{2a001323-07d6-31c1-910a-be0b0a4eecf7}). In such settings it may remain plausible that measurement error is independent of the latent variables, but it is typically unrealistic to assume the measurement errors themselves are uncorrelated. The composite‑error specification accommodates these scenarios: if the measurement‑error vector is jointly centrally symmetric (so its third‑order joint cumulants vanish) and independent of the latent component, and if there exist skewed equation shifters, then our framework continues to apply.

Another potential scenario is omitted/unobservable effects in a simultaneous equation system. Omitted variable bias can occur if certain causal variables cannot be measured (e.g., latent ability in the returns‑to‑education literature) or if the researcher prefers to specify a smaller model. One place such model misspecification may arise is in the VAR literature, where the sample size is often too limited to estimate big systems. If the omitted variables can be modeled as symmetric, then the framework here applies, and the estimation and identification arguments in this paper offer a method robust to such omission.
\end{itemize}


\section{Identification Concept: Identification Up to a Permutation and Scaling}
\label{sec:identification_concept}

Consistent with most studies that leverage higher‑order cumulant structures for identification, Assumptions~A1--A3 at best ensure the identification of \(A\) (and hence \(\Lambda\)) only up to a permutation and scaling. A familiar illustration of this concept appears in the Gaussian VAR literature. Specifically, if one assumes independence across structural errors and imposes a triangular (or acyclic) system, then the Cholesky decomposition of the variance matrix can identify the contemporaneous interaction matrix only up to a permutation and scaling. \footnote{%
Formally, conditional on a chosen ordering, the standard Cholesky factor with positive diagonal
fixes a particular scale. In applied work it is common to re‑normalize (e.g., setting
\(\operatorname{diag}(\Lambda)=\mathbf{1}\) or shock variances to unity) after the Cholesky step.
Our discussion treats such conventional re‑scalings as part of the diagonal‑scaling indeterminacy.} Concretely:

\begin{definition}
    The mixing matrix \(A\) (or, equivalently, the structural parameter matrix \(\Lambda\)) is said to be \emph{identified up to a column permutation and scaling} if \(\widetilde{A} = A D P\) (or equivalently \(\widetilde{\Lambda} = P^{-1} D^{-1} \Lambda \)) can be identified, where \(D\) is a diagonal matrix and \(P\) is a permutation matrix.
\end{definition}

Scale indeterminacy is a routine and generally minor issue in linear models; it is typically addressed by normalizing the diagonal of \(\Lambda\) to unity. While this normalization resolves scale indeterminacy, the situation becomes more intricate when permutation indeterminacy is present. Overcoming permutation indeterminacy is often the principal hurdle between the aforementioned identification result and point identification. Fortunately, economic theory can be instrumental in pinpointing the correct permutation. Once an economically meaningful permutation or labeling is established, a diagonal normalization can then be used to eliminate scale indeterminacy.

Although resolving permutation indeterminacy is not the central focus of this paper, we outline two practical examples of how one might obtain point identification in structural models, guided by economic reasoning:

\begin{itemize}
    \item [P1] \textbf{Sign Restrictions:} 
    Define an auxiliary matrix \(\Lambda^{\text{sign}}\) to capture the sign pattern of \(\Lambda\). Its entries are given by
    \begin{equation}
        \Lambda^{\text{sign}}_{ij} = 
        \begin{cases}
            1, & \text{if }\Lambda_{ij} > 0,\\
            -1, & \text{if }\Lambda_{ij} < 0,\\
            0, & \text{if }\Lambda_{ij} = 0.
        \end{cases}
    \end{equation}
    Adopt the normalization \(\operatorname{diag}(\Lambda)=\mathbf{1}\) (so row signs are fixed). If \(\Lambda^{\text{sign}}\) has pairwise distinct rows, then one can identify a unique permutation consistent with \(\Lambda^{\text{sign}}\). This allows for a natural labeling of the structural equations, reflecting each equation's unique sign pattern. For instance, in a classic two-equation supply-and-demand system, the supply curve is upward-sloping whereas the demand curve is downward-sloping. Hence, labeling and ordering follow directly from these sign patterns.
    
    \item [P2] \textbf{Causal Ordering/Acyclicality:} 
    A common economic restriction is acyclicity, under which there exists a permutation of the
variables that renders the structural matrix triangular with nonzero diagonal (a recursive system).
When such a triangular representation exists and satisfies the usual \emph{faithfulness} condition—
namely, there is a \emph{unique} permutation \(P^\star\) for which \(P^\star \Lambda\) is triangular
with the implied zero pattern, and no other permutation yields the same triangular zero pattern
(i.e., no exact cancellations create spurious zeros)—the causal ordering is uniquely determined.
In large samples, reordering the estimated system to achieve this triangular zero pattern recovers
\(P^\star\) with probability approaching one, delivering point identification of the ordering; see,
for example, the LiNGAM approach of \cite{shimizu2006linear}.
\end{itemize}

A crucial advantage of these procedures is that, once a consistent estimator of \(\widetilde{A}\) is available, one can choose the correct permutation of its columns with probability tending to 1 as the sample size goes to infinity after estimation. In \cite{lewis2021identifying}, rules like this are referred to as \emph{consistent labeling criteria}. Theorem~4 in \cite{lewis2021identifying} shows that the error introduced by the post‑estimation labeling step is asymptotically negligible: the permutation choice does not alter the first‑order asymptotic distribution of the estimator. Consequently, once an asymptotically normal estimator of the permuted matrix is available, standard delta‑method arguments yield valid inference for normalized structural parameters.

Throughout this paper, we focus on settings analogous to P1 and P2, in which
economic theory supplies a \emph{consistent‑labeling criterion} that eliminates
both scale and permutation indeterminacies.  
In this context, once a consistent and asymptotically normal estimator of the
permuted structural‑parameter matrix \(\widetilde{A}\) is available, inference
on the structural parameters follows directly from Theorem 4 in
\cite{lewis2021identifying}.  
Consequently, our objective is to show that, under
Assumptions A1–A3, the mixing matrix \(A\) is identified up to a column
permutation and scaling, and that consistent, asymptotically normal estimators
of \(\widetilde{\Lambda}\) can be constructed.


\section{Identification}

In this section, we state the main identification theorem under the baseline model for $h=3$. Although the argument extends to all $h>2$, the general proof differs from the $h=3$ case only by minor technical details. For clarity and brevity, we relegate the general proof to the Appendix. We then discuss some useful extensions of the baseline model.

\subsection{Identification Using Third Cumulant Information Alone}

Let $X_c := X - \mathbb{E}[X]$ denote the centered observables. By translation invariance of cumulants (orders $\ge 2$), $\kappa_3(w^\top X)=\kappa_3(w^\top X_c)$. Let $w$ be a $d$-dimensional column vector and consider the objective
\begin{equation}
   Q(w) \;=\; \kappa_3\!\bigl(w^\top X\bigr)
   \;=\; \mathbb{E}\Bigl[\bigl(w^\top X_c\bigr)^3\Bigr].
\end{equation}
We begin with a lemma concerning the Hessian of this objective function.

\begin{lemma}
\label{lemma:hessian}
\textit{Under Assumptions A1$^\prime$--A3$^\prime$ and $\mathbb E\|X\|^3<\infty$,}
We have
\begin{equation}
    \nabla_w^2 Q(w) \;=\; A \,D_w \,A^\top,
\end{equation}
where $D_w$ is a diagonal matrix whose entries depend on $w$.
\end{lemma}

\begin{proof}
First, write $X_c = A S_c$ with $S_c := S - \mathbb{E}[S]$. Since third cumulants are translation invariant, $\kappa_3(S_i)=\kappa_3((S_c)_i)$ and $\kappa_3(w^\top X)=\kappa_3(w^\top X_c)$.
We simplify the objective:
\begin{equation}
\begin{aligned}
    Q(w) 
    &= \kappa_3\!\bigl(w^\top X_c\bigr) \\
    &= \kappa_3\!\bigl(v^\top S_c\bigr) 
    \quad \text{where } v^\top = w^\top A, \\
    &= \sum_{q=1}^{d} v_q^{\,3} \,\kappa_3(S_q),
    \quad \text{(eq (\ref{eq:Homogeneity}),eq (\ref{eq:separability}), A3$^\prime$).}
\end{aligned}
\end{equation}
Next, take the Hessian with respect to $v$. Since each term involves only a single $v_q$,
\[
\nabla^2 Q(v)_{ij} 
\;=\; 
\begin{cases}
    6\,v_i \,\kappa_3(S_i), & \text{if } i = j,\\[6pt]
    0, & \text{if } i \neq j.
\end{cases}
\]
Hence $\nabla^2 Q(v)$ is diagonal, with diagonal entries $6\,v_i\,\kappa_3(S_i)$.

Finally, recall that $v = A^\top w$. By the chain rule,
\[
\nabla_w^2 Q(w) \;=\; A \,\bigl(\nabla_v^2 Q(v)\bigr)\,A^\top 
\;\equiv\;
A \,D_w\,A^\top,
\]
where $D_w = \nabla_v^2 Q(v)$ is diagonal. This completes the proof.
\end{proof}

We now state and prove our main identification theorem:

\begin{theorem}\label{thm:eig-decomp}
Consider
\[
H \;=\; \nabla_w^2 Q(w_2)^{-1}\,\nabla_w^2 Q(w_1),
\]
with $w_1\neq w_2$. Under Assumptions A1$^\prime$--A3$^\prime$ and the condition
\[
(A^\top w_2)_i \;\neq\; 0 \quad \text{for all } i,
\]
the rows of $\Lambda$ are eigenvectors of $H$ for every $w_1$. 

In particular, if we fix $w_2=\mathbf{1}$ and assume $\sum_q A_{q i}\neq 0$ for all $i$ (so that $(A^\top w_2)_i=\sum_q A_{q i}\neq 0$), then the above condition holds. Furthermore, with this choice of $w_2$, if $w_1$ is drawn uniformly at random from the $d$-dimensional unit cube, then $H$ has (almost surely) distinct eigenvalues, and its eigenvectors are the columns of $\Lambda^\top$ up to scaling (and permutation).
\end{theorem} 

\begin{proof}
First, write
\begin{equation}\label{eq:H_decomp}
\begin{aligned}
    H 
    &= \nabla_w^2 Q(w_2)^{-1}\,\nabla_w^2 Q(w_1) \\
    &= (A^\top)^{-1} \,D_{w_2}^{-1}\,A^{-1}\,\bigl(A\,D_{w_1}\,A^\top\bigr) \\
    &= (A^\top)^{-1} \,D_{w_2}^{-1}\,D_{w_1}\,A^\top, 
\end{aligned}
\end{equation}
where $D_{w_1}$ and $D_{w_2}$ are diagonal matrices arising from $\nabla_w^2 Q(w_1)$ and $\nabla_w^2 Q(w_2)$, respectively (by Lemma~\ref{lemma:hessian}).

Since $\Lambda = A^{-1}$, multiply appropriately to obtain
\[
H 
\;=\; 
\Lambda^\top \bigl(D_{w_2}^{-1}\,D_{w_1}\bigr) (\Lambda^\top)^{-1}.
\]
This yields an eigen-decomposition of $H$. In particular, 
$\tfrac{D_{w_1,\,ii}}{D_{w_2,\,ii}}$ is the $i$th eigenvalue, and the columns of $\Lambda^\top$ (equivalently, the rows of $\Lambda$) are the corresponding eigenvectors.

Next, fix $w_2=\mathbf{1}$ and write $v_\ell:=A^\top w_\ell$ for $\ell\in\{1,2\}$. 
Then $v_{2,i}=(A^\top w_2)_i=\sum_{q} A_{q i}$. Since 
$D_w=\operatorname{diag}\big(6\,\kappa_3(S_i)\,v_i\big)$, we have
\[
D_{w_2}
= \operatorname{diag}\big(6\,\kappa_3(S_i)\,v_{2,i}\big),
\qquad
D_{w_2}^{-1}
= \operatorname{diag}\Big(\frac{1}{6\,\kappa_3(S_i)\,v_{2,i}}\Big).
\]
Hence
\[
D_{w_2}^{-1}D_{w_1}
= \operatorname{diag}\!\Big(\frac{v_{1,i}}{v_{2,i}}\Big),
\quad\text{i.e.,}\quad
(D_{w_2}^{-1}D_{w_1})_{ii}
= \frac{(A^\top w_1)_i}{(A^\top \mathbf{1})_i}
= \frac{(w_1^\top A)_i}{\sum_{q} A_{q i}}.
\]

To show these diagonal entries (the eigenvalues) are almost surely distinct for a random $w_1$, observe that the condition 
\[
\frac{(w_1^\top A)_i}{\sum_{q} A_{q i}}
=
\frac{(w_1^\top A)_j}{\sum_{q} A_{q j}}
\]
is equivalent to
\begin{equation}\label{eq:eig_equal}
   w_1^\top 
   \Bigl(
   A_{\cdot i}\,\tfrac{1}{\sum_{q} A_{q i}}
   \;-\;
   A_{\cdot j}\,\tfrac{1}{\sum_{q} A_{q j}}
   \Bigr)
   \;=\; 0.
\end{equation}
Since $A$ is full rank and $\sum_{q} A_{q i} \neq 0$ for all $i$, the vector 
$A_{\cdot i}/\sum_{q} A_{q i} - A_{\cdot j}/\sum_{q} A_{q j}$ is nonzero whenever $i\neq j$. Hence, \eqref{eq:eig_equal} defines a $(d-1)$-dimensional hyperplane in $\mathbb{R}^d$ with Lebesgue measure zero. A uniformly random $w_1$ on the unit cube hits finitely many such hyperplanes with probability zero. Consequently, with probability one, no two eigenvalues coincide.

Finally, since each distinct eigenvalue has a one-dimensional eigenspace, the matrix of eigenvectors is unique up to scaling and permutation of columns. In other words, the recovered matrix of eigenvectors identifies $\Lambda^\top$ up to permutation and column-scaling (e.g., normalizing each eigenvector to have unit length).
\end{proof}

Working in the coordinates of the latent variable \(S\) yields straightforward proofs. However, the crucial role of Lemma \ref{lemma:hessian} can be obscured by this simplicity. Its intuition is more transparent when viewed in the coordinates of the observed \(X\) through tensor contraction.

For centered data, the third-order cumulant tensor \(C_3(X_c)\) has entries
\[
\big[C_3(X_c)\big]_{ijk}=\mathbb{E}\!\big[(X_c)_i (X_c)_j (X_c)_k\big],
\]
so for any \(w\in\mathbb{R}^d\), the expression \(C_3(X_c)[w,\cdot,\cdot]\) denotes its mode-1 contraction with \(w\), i.e., the \(d\times d\) matrix
\[
\big[C_3(X_c)[w,\cdot,\cdot]\big]_{jk}
=\sum_{i=1}^d w_i\,\big[C_3(X_c)\big]_{ijk}
=\mathbb{E}\!\big[(w^{\top}X_c)\,(X_c)_j\,(X_c)_k\big].
\]
Intuitively, we can view the third-order cumulant tensor as slices of \(d\times d\) matrices indexed by \(i\), and the contraction takes a linear combination of these slices to form a weighted sum, with \(w_i\) the weight on the \(i\)th slice.

From a straightforward application of the chain rule in the coordinates of \(X\), the Hessian
\[
\nabla_w^2 Q(w)=6\,\mathbb{E}\!\big[(w^{\top}X_c)\,X_cX_c^{\top}\big]
=6\,\big[C_3(X_c)[w,\cdot,\cdot]\big],
\]
is exactly this tensor-contraction operation with weight \(w\).

In the latent coordinates, \(X_c=AS_c\) and \(w^{\top}X_c=v^{\top}S_c\) with \(v=A^{\top}w\), so
\[
\mathbb{E}\!\big[(w^{\top}X_c)\,X_cX_c^{\top}\big]
= A\,\mathbb{E}\!\big[(v^{\top}S_c)\,S_cS_c^{\top}\big]A^{\top}
= A\Big(\sum_{i=1}^d v_i\,\mathbb{E}[\,S_{ci}\,S_cS_c^{\top}\,]\Big)A^{\top}.
\]
Thus, contracting the third-order cumulant tensor for \(X\) is, up to a common congruence by \(A\), equivalent to contracting \(C_3(S_c)\) along the vector \(v\).

Because cumulants are multilinear and translation-invariant, under the diagonality of \(C_3(S)\) (Assumption A3$^\prime$), any contracted matrix \(C_3(S_c)[v,\cdot,\cdot]\) is diagonal, with diagonal entries proportional to \(\kappa_3(S_i)\,v_i\). Lifting back to the coordinates of the observables applies a congruence transformation by \(A\) to these diagonal matrices:

\[
\nabla_w^2 Q(w)=A\,\mathrm{diag}\!\big(6\,\kappa_3(S_i)\,v_i\big)\,A^{\top}\equiv A\,D_wA^{\top},
\]
where the weights depend on the contraction vector \(w\) (equivalently, \(v=A^{\top}w\)). This Hessian device gives fine control over the collapse: the diagonal weights inherit the variation in \(w\), so different directions produce different diagonal scalings.

The main theorem then observes that two such contractions—at \(w_1\) and \(w_2\)—yield \(G(w_\ell)\coloneqq A D_{w_\ell}A^{\top}\) with a common congruence. Forming
\[
H=G(w_2)^{-1}G(w_1)=\Lambda^{\top}D_{w_2}^{-1}D_{w_1}(\Lambda^{\top})^{-1}
\]
removes the congruence and produces a matrix similar to a diagonal one; its eigenvectors are exactly the rows of \(\Lambda=A^{-1}\), up to scaling and permutation. Choosing \(w_1\) from a distribution with a density ensures distinct eigenvalues almost surely, guaranteeing uniqueness of the eigendirections.

The above argument is inspired by the joint diagonalization used by \citet{bonhomme2009consistent} and ICA joint-diagonalization variants, but it follows a different and simpler path. Identification here requires only two matrices, and—by Lemma~\ref{lemma:hessian}—both arise solely from the third cumulant (or any fixed order \(h>2\)) via tensor contraction, so no fourth- or second‑order information is required, in contrast to \citet{bonhomme2009consistent}. Consequently, there is no need to assemble large families of cumulant matrices, introduce off‑diagonal minimization objectives, or impose second‑moment diagonality to enable pre‑whitening \citep{cardoso1993blind,hyvarinen2013independent}. Identification is also robust to non‑diagonality in other cumulants, and when researchers do wish to maintain assumptions on those cumulants (e.g., in SVAR applications), the restriction becomes a testable implication in our framework rather than a maintained assumption. Relative to non‑orthogonal JD methods that also avoid whitening, our construction does not rely on fixed‑point iterations \citep{1011195}, tuning‑intensive off‑diagonality criteria \citep{4671095}, or positive‑definiteness constraints tied to the covariance \citep{10.5555/1005332.1016784}. As we will show in Section~6, this simplicity leads to transparent large‑sample analysis for the sample‑analogue estimator.

\subsection{Important Extensions to the Baseline Model}

The baseline identification result relies on two conditions:
(i) every structural error has a nonzero third cumulant, and
(ii) the structural parameter matrix $A$ is square and invertible.
Some popular empirical settings in economics and statistics violate one or
both of these assumptions. Throughout this subsection we continue to work with the centered observables $X_c:=X-\mathbb{E}[X]$ and the cumulant objective
$Q(w)=\kappa_3(w^\top X)=\mathbb{E}[(w^\top X_c)^3]$, so that
\[
G(w):=\nabla_w^{2}Q(w)=A D_w A^\top,
\qquad
D_w:=\operatorname{diag}\!\big(6\,\kappa_3(S_i)\,(A^\top w)_i\big),
\]
and we let $X\in\mathbb{R}^{d_1}$, $S\in\mathbb{R}^{d_2}$, $A\in\mathbb{R}^{d_1\times d_2}$ (so $d_1=d_2=d$ in the baseline case). As shown before, by translation invariance of cumulants (orders $\ge 2$), working with $X_c$ is without loss for identification.

For example, in independent component analysis (ICA) and other
higher‑order‑cumulant methods it is common to allow one (or several) structural
errors to exhibit \emph{zero} skewness, in which case
$G(w)=A D_w A^{\top}$ is singular for all $w$ and ordinary inversion is unavailable.
Moreover, in many economic applications, even when only a subset of structural errors exhibits nonzero skewness, it is still desirable for the model to deliver partial information about the structural parameters (e.g., for the columns of $A$ associated with the skewed shocks).
On the other hand, in factor models the loading matrix $A$ is typically
\emph{tall} ($d_{1}>d_{2}$), creating a similar noninvertibility problem because $\operatorname{rank}(G(w))\le d_2<d_1$ for all $w$.

By contrast, non‑Gaussian VAR applications often satisfy the baseline
invertibility and nonzero skewness conditions, yet practitioners also impose the
additional assumption that structural shocks are uncorrelated.
This second‑moment information is valuable and can be exploited whenever
available via $\mathrm{Var}(X)=A D_2 A^\top$ with diagonal $D_2$.

The remainder of this subsection develops three extensions that accommodate
these practical scenarios:
\begin{enumerate}
  \item Extending the diagonalization argument to non‑square loading matrices
  \item Extending identification to models with non-skewed structural errors.
  \item Leveraging the uncorrelated‑shocks assumption to construct an
        over‑identification test.
\end{enumerate}

\paragraph{Extension 1: Nonsquare $A$ with skewed factors}

A common setting for factor models is when the dimension of the observed variables exceeds the dimension of the latent factors. Concretely, let \(X=A\,S\) with \(A\in\mathbb{R}^{d_1\times d_2}\), \(d_1>d_2\), and \(A\) full column rank. Under A2\(^{\prime}\)–A3\(^{\prime}\) and for choices of \(w_2\) such that \((A^\top w_2)_i\neq 0\) for all \(i\) (e.g., \(w_2=\mathbf{1}\) with \(\sum_q A_{qi}\neq 0\) for all \(i\)), we have for any \(w\):
\[
G(w):=\nabla_w^2 Q(w)=A D_w A^\top,
\qquad
D_w:=\operatorname{diag}\!\big(6\,\kappa_3(S_i)\,(A^\top w)_i\big).
\]
In the tall case (full column rank \(A\) and invertible \(D_{w_2}\)),
\[
G(w_2)^{+}=(A D_{w_2} A^\top)^{+}=(A^+)^{\!\top}\,D_{w_2}^{-1}\,A^{+},
\]
where $A^+$ denotes the Moore–Penrose inverse. Therefore
\[
H \;:=\; G(w_1)\,G(w_2)^{+}
= A\,\big(D_{w_1}D_{w_2}^{-1}\big)\,A^{+}
= A\,\Delta\,A^{+},
\qquad
\Delta=\operatorname{diag}\!\Big(\tfrac{(A^\top w_1)_i}{(A^\top w_2)_i}\Big)_{i=1}^{d_2}.
\]
It follows that
\[
H\,A_{\cdot i}=\delta_i\,A_{\cdot i},\qquad \delta_i=\Delta_{ii},
\]
so the \(d_2\) \emph{right} eigenvectors of \(H\) associated with its nonzero eigenvalues are precisely the columns of \(A\), identified up to column scaling and permutation. The remaining \(d_1-d_2\) eigenvalues equal zero, with eigenvectors spanning \(\operatorname{col}(A)^\perp\). As in the main theorem, if \(w_1\) is drawn from any absolutely continuous distribution on \(\mathbb{R}^{d_1}\) (with \(w_2\) fixed as above), the \(\{\delta_i\}\) are almost surely pairwise distinct.

\paragraph{Extension 2: Identification when there are non‑skewed structural errors}

Suppose some components of the structural error have zero third cumulant. Let \(J:=\{i:\kappa_3(S_i)\neq 0\}\) and assume \(J\neq\varnothing\). Then, for every \(w\), the diagonal matrix \(D_w=\operatorname{diag}\!\big(6\,\kappa_3(S_i)\,(A^\top w)_i\big)\) has zeros on the coordinates \(J^c\), so \(G(w)=A D_w A^\top\) is singular and the ordinary inverse used in Theorem~\ref{thm:eig-decomp} is unavailable.

In this case, what is identifiable depends on the target parameter. If our parameter of interest is the \emph{mixing matrix} \(A\), we lose very little. For \(\ell=1,2\) we can write
\[
G(w_\ell):=A_J D_{\ell,J} A_J^\top,
\qquad
D_{\ell,J}:=\operatorname{diag}(d_{\ell i})_{i\in J},
\quad
d_{\ell i}:=6\,\kappa_3(S_i)\,(A^\top w_\ell)_i,
\]
where \(A_J\) collects the columns of \(A\) indexed by \(J\). Assuming \((A^\top w_2)_i\neq 0\) for all \(i\in J\) (so \(D_{w_2,J}\) is invertible), we are \emph{back in the tall case} of Extension~1 with \(A_J\in\mathbb{R}^{d_1\times |J|}\) full column rank. Hence
\[
H_{\text{mix}} \;:=\; G(w_1)\,G(w_2)^{+}
= A_J\,\big(D_{w_1,J}D_{w_2,J}^{-1}\big)\,A_J^{+},
\]
and the \( |J| \) right eigenvectors corresponding to its nonzero eigenvalues are exactly the columns \(\{A_{\cdot i}: i\in J\}\), identified up to column scaling and permutation (with almost‑sure eigenvalue distinctness for random \(w_1\)).

By contrast, identification of the \emph{structural parameter matrix} \(\Lambda=A^{-1}\) is generally not possible from third‑cumulant information alone when \(|J|<d_2\). The third‑order cumulant family \(\{G(w):w\}\) depends only on \(A_J\) and is invariant to how the complementary columns \(A_{J^c}\) are chosen (as long as \(A=[A_J\;A_{J^c}]\) has full column rank). Different completions \(A_{J^c}\) lead to different inverses \(\Lambda\) but the same \(\{G(w)\}\). Thus, without additional structure (e.g. second‑moment restrictions such as uncorrelated structural errors), the rows of \(\Lambda\) cannot be identified when some components are non‑skewed.

\paragraph{Extension 3: Uncorrelated structural errors and an overidentification test}

In some econometric applications one may wish to impose that the structural errors are \emph{uncorrelated}. Writing $X_c:=X-\mathbb{E}[X]$, this means the second cumulant (covariance) satisfies
\[
\mathrm{Var}(X)\;=\;\mathbb{E}[X_c X_c^\top]\;=\;A\,D_2\,A^\top,
\]
where $D_2$ is diagonal with the structural variances on its diagonal.

\smallskip
\noindent\textit{Identification remark (optional).}
If one were to exploit the second‑moment structure in estimation, one could replace $\bigl(\nabla_w^2 Q(w_2)\bigr)^{-1}$ by $\mathrm{Var}(X)^{-1}$ in Theorem~\ref{thm:eig-decomp}, yielding
\[
H_\Sigma \;:=\; \mathrm{Var}(X)^{-1}\,\nabla_w^2 Q(w_1)
\;=\;\Lambda^\top\bigl(D_2^{-1}D_{w_1}\bigr)(\Lambda^\top)^{-1},
\]
so the eigenvectors are again the rows of $\Lambda$ (up to scaling/permutation).Note that with a diagonal second cumulant, the non-inverse problem does not arise, and hence we can recover all rows of $\Lambda$ corresponding to the distinct eigenvalues.  \emph{For the test below, however, we deliberately estimate the demixing matrix using third‑cumulant information only}, so that the second‑moment implication can be tested rather than imposed.

\smallskip
\noindent\textit{Testable implication (Testing joint diagonality of the second and the third cumulants).}
Let $\widetilde{\Lambda}$ denote the (scaled‑and‑permuted) demixing matrix recovered \emph{solely} from third‑cumulant diagonalization, i.e.
\[
\widetilde{\Lambda} \equiv \text{(rows of eigenvectors of }H:=(\nabla_w^2 Q(w_2))^{-1}\nabla_w^2 Q(w_1)\text{, oriented and normalized)}.
\]
Under A1$^\prime$--A3$^\prime$ plus uncorrelated structural errors,
\[
\widetilde{\Lambda}\,\mathrm{Var}(X)\,\widetilde{\Lambda}^\top
\;=\; P^* D^*\,\Lambda\,\mathrm{Var}(X)\,\Lambda^\top D^* P^{*\,\top}
\;=\; P^* D^*\,D_2\,D^* P^{*\,\top}
\]
is \emph{diagonal}. Equivalently, its \emph{unique} off‑diagonal entries are zero in population, regardless of the unknown row scaling and permutation. Let $\vech_{\mathrm{off}}(\cdot)$ denote the operator that stacks the strict upper‑triangular entries of a symmetric matrix into a vector in $\mathbb{R}^{\binom{d}{2}}$. The overidentifying restrictions are
\[
\vech_{\mathrm{off}}\!\Bigl(\,\widetilde{\Lambda}\,\mathrm{Var}(X)\,\widetilde{\Lambda}^\top\,\Bigr)\;=\;0\in\mathbb{R}^{\binom{d}{2}}.
\]

\smallskip
\noindent\textit{Sample statistic and asymptotics (outline).}
Estimate $\widetilde{\Lambda}$ from third‑order information only by
\[
\hat H \;=\; \bigl(\nabla_w^2 \hat Q(w_2)\bigr)^{-1}\,\nabla_w^2 \hat Q(w_1),
\quad
\hat{\widetilde{\Lambda}}=\text{(oriented, normalized rows of eigenvectors of }\hat H\text{)},
\]
and estimate the covariance by the centered sample covariance
\[
\widehat{\Sigma}_X \;=\; \frac{1}{n}\sum_{i=1}^n (X_i-\bar X)(X_i-\bar X)^\top.
\]
Define the $\binom{d}{2}\times 1$ vector of unique off‑diagonals
\[
\hat r \;=\; \vech_{\mathrm{off}}\!\Bigl(\,\hat{\widetilde{\Lambda}}\,\widehat{\Sigma}_X\,\hat{\widetilde{\Lambda}}^\top\,\Bigr).
\]
If an asymptotically normal estimator $\hat{\widetilde{\Lambda}}$ is available (as established in the Section 6; see Theorem~\ref{thm:main_sample_analogue_theorem}), then, under mild moment conditions, a joint Delta‑method argument implies
\[
\sqrt{n}\,\hat r \;\xrightarrow{d}\; \mathcal{N}(0,\Omega),
\]
for some positive definite $\Omega$ that can be consistently estimated (e.g., via a plug‑in Delta method based on raw moments up to order~6, or via jackknife/bootstrap).
Consequently, the quadratic form
\[
T_n \;=\; n\,\hat r^\top\,\hat\Omega^{-1}\,\hat r
\]
is asymptotically $\chi^2$ with $\binom{d}{2}$ degrees of freedom under the null.
The construction is invariant to row scaling and permutation of $\hat{\widetilde{\Lambda}}$, so no additional normalization or row matching is required. Full details—covariance construction and the proof that $T_n\Rightarrow\chi^2_{\binom{d}{2}}$—are provided in the Appendix. Moreover, we show that under the usual regularity conditions for VAR models (e.g., those adopted in \citet{DAVIS2023180}), the test statistic is asymptotically unaffected by the first-stage OLS estimation used to partial out the lags.

\section{Estimation}
\noindent
Given the constructive nature of the identification proof, a plug-in estimator is a natural choice.
In this section, we establish the asymptotic properties of the simple plug-in estimator under the baseline model by showing that it is a smooth function of sample averages.
Estimation procedures for the model extensions discussed in Section~5 are provided in the appendix.

\subsection{Population Setup}

Recall the population objective function (defined via the third \emph{cumulant})
\[
Q(w) \;=\; \kappa_3\!\bigl(w^\top X\bigr)
\;=\; \mathbb{E}\Bigl[\bigl(w^\top X_c\bigr)^3\Bigr],
\qquad
X_c := X - \mathbb{E}[X].
\]
It is straightforward to see that the entries of its Hessian matrix are a linear combination of the entries of the third-order cumulant tensor; for order three, cumulants coincide with centered moments:
\[
\nabla_w^2 Q(w)
\;=\; 6\,\mathbb{E}\!\big[(w^\top X_c)\,X_c X_c^\top\big]
\;=\; 6\sum_{r=1}^d w_r\,\mathbb{E}\!\big[(X_c)_r\,X_c\,X_c^\top\big].
\]

For asymptotics it is convenient to parameterize everything by \emph{raw} moments and apply the \emph{cumulant map} inside the moment-to-Hessian construction.
Let
\[
M(X)\in\mathbb{R}^{D_3(d)-1}
\quad\text{stack all raw monomials up to total degree 3 (excluding the constant),}
\]
so that \(D_3(d)-1=\binom{d+3}{3}-1\).
For example, when \(d=2\),
\[
M(X)^\top
= \bigl(X_1,\,X_2,\,X_1^2,\,X_1X_2,\,X_2^2,\,X_1^3,\,X_1^2X_2,\,X_1X_2^2,\,X_2^3\bigr).
\]

We then define the \emph{cumulant map} (still denoted \(C\))
\[
C:\mathbb{R}^{D_3(d)-1}\to\mathbb{R}^{d_3},\qquad
c_3 := C(M),
\]
where \(d_3=\binom{d+2}{3}\) is the number of distinct entries in the (symmetric) third-order cumulant tensor.
\emph{Intuitively, \(C\) takes raw moments and, via the standard moment–cumulant identities, subtracts off lower-order products to return third-order cumulants (which, at order three, equal centered third moments).}\footnote{For any fixed order \(r\), the map from raw moments up to order \(r\) to cumulants up to order \(r\) is a multivariate polynomial}

The Hessian \(\nabla_w^2 Q(w)\) is an \emph{affine} function of \(c_3\), so the population matrix
\[
H(M)
\;=\;
\bigl(\nabla_w^2 Q(w_2)\bigr)^{-1}\,\nabla_w^2 Q(w_1)
\;=\;
H\!\big(C(M)\big)
\]
is an analytic function of \(M\) on any neighborhood where \(\nabla_w^2 Q(w_2)\) is nonsingular.
Order the eigenvalues of \(H(M)\) in decreasing order and let \(u_k(M)\) denote the associated normalized eigenvector, oriented by one of the following conventions.

\medskip
\noindent\emph{Orientation A (row-sum rule).} Assume each row of \(\Lambda\) has nonzero sum, i.e.\ \(\mathbf{1}^\top \Lambda_{i\cdot}\neq 0\) for all \(i\). Fix the sign by requiring
\[
\mathbf{1}^\top u_k(M) \;>\; 0 .
\]

\noindent\emph{Orientation B (largest-entry rule).} Alternatively, assume that for each \(k\) the largest absolute coordinate of the \emph{population} eigenvector \(u_k(M)\) is unique. Let \(j_k=\arg\max_j |[u_k(M)]_j|\) and fix the sign by requiring \([u_k(M)]_{j_k}>0\).

\medskip
Either convention yields a well-defined, smooth map
\[
g:\mathbb{R}^{D_3(d)-1}\to\mathbb{R}^d,\qquad g(M)=u_k(M),
\]
in a neighborhood of the population \(M=\mathbb{E}[M(X)]\).

\subsection{Sample Counterparts and Estimator Construction}

For the \emph{point estimator} we use the natural centered sample analogue of the Hessian:
\[
\nabla_w^2 \hat{Q}(w) 
\;=\; 
6\,\frac{1}{n}\sum_{i=1}^n \bigl(w^\top X_{c,i}\bigr)\,X_{c,i} X_{c,i}^\top,
\qquad
X_{c,i}:=X_i-\bar X,\ \bar X:=\tfrac1n\sum_{i=1}^n X_i .
\]
Equivalently, one may compute \(\nabla_w^2 \hat{Q}(w)\) by first forming the vector of sample \emph{raw} moments \(\hat M=\tfrac{1}{n}\sum_{i=1}^n M(X_i)\), then applying the cumulant map \(C\) to obtain the sample third cumulants \(C(\hat M)\), and finally plugging these into the affine formula for \(\nabla_w^2 Q(w)\); the two procedures coincide algebraically because third-order cumulants equal centered third moments.

Define the sample analogue of \(H\) by
\[
\hat H
\;=\;
\bigl(\nabla_w^2 \hat{Q}(w_2)\bigr)^{-1}\,\nabla_w^2 \hat{Q}(w_1).
\]
Let \(\tilde u_k\) denote a (possibly complex-valued) normalized eigenvector of \(\hat H\) associated with its \(k\)th largest eigenvalue, oriented by the \emph{same} rule used in population (Orientation~A or~B, applied after normalization).
We then define our estimator as the \emph{real part}
\[
\hat{u}_k \;=\; \Re\bigl(\tilde u_k\bigr).
\]
This “real-part” safeguard is asymptotically inactive: at the population \(M\), the relevant eigenvalue is real and simple (Theorem~\ref{thm:eig-decomp}), so there exists a neighborhood of \(M\) on which the eigenvector map is real-analytic and real-valued; within that neighborhood \(\Re(\tilde u_k)=\tilde u_k\). Moreover, \(\Re(\cdot)\) is a real-linear projection \(\mathbb{C}^d\to\mathbb{R}^d\), so composing with \(\Re\) preserves differentiability at \(M\).

For asymptotics we view \(u_k\) as a function of the \emph{raw} moment parameter \(M\) and write \(g(M)=u_k(M)\). Our estimator is the plug-in
\[
\hat{u}_k \;=\; g(\hat{M}),
\qquad
\hat M=\frac{1}{n}\sum_{i=1}^n M(X_i).
\]

\begin{lemma}\label{lemma:g_diff}
$g$ is differentiable at $M$ \emph{provided the relevant eigenvalue of $H\!\big(C(M)\big)$ is simple (e.g., when $w_1$ is drawn from a distribution with a density and $(A^\top w_2)_i\neq 0$ for all $i$).}
\end{lemma}

\begin{proof}
$C$ is a polynomial map (moment–cumulant relations at any fixed order are polynomial), hence analytic. When the eigenvalue of interest is simple, the associated normalized, oriented eigenvector depends analytically on the entries of $H$; see, e.g., standard eigenvector perturbation results (See \citet{stewart1990matrix}, \citet{kato1995perturbation}, \citet{371a9e14-3dd5-3f4a-b3ef-88e02cbdde26}). Hence $g=g_{\mathrm{eig}}\circ C$ is differentiable at $M$. Since $\Re$ is real-linear, $\Re\circ g$ is also differentiable at $M$.
\end{proof}

From the Continuous Mapping Theorem and the Delta Method, we obtain the main estimation result:

\begin{theorem}[Asymptotic linearity, variance, and CLT for $\hat u_k$ with $w_2=\mathbf{1}$]
\label{thm:main_sample_analogue_theorem}
Assume A1$^\prime$--A3$^\prime$ and that $\sum_q A_{q i}\neq 0$ for all $i$ (so that $(A^\top \mathbf{1})_i\neq 0$).
Fix $w_2=\mathbf{1}$ and draw $w_1$ from an absolutely continuous distribution on $\mathbb{R}^d$, independent of the sample.
Let $M=\mathbb{E}[M(X)]$ be the vector of raw moments up to total degree 3 (excluding the constant) and $\hat M=\tfrac{1}{n}\sum_{i=1}^n M(X_i)$ its sample analogue.
Assume $\{X_i\}_{i=1}^n$ are i.i.d.\ with $\mathbb{E}\|X\|^{6}<\infty$.
Define
\[
H(M)=\bigl(\nabla_w^2 Q(w_2)\bigr)^{-1}\nabla_w^2 Q(w_1)
=\;H\!\big(C(M)\big),
\qquad
\hat H=\bigl(\nabla_w^2 \hat Q(w_2)\bigr)^{-1}\nabla_w^2 \hat Q(w_1).
\]
Then, with probability one over the draw of $w_1$, $H(M)$ has simple (pairwise distinct) eigenvalues, $u_k$ is well defined, and the map $g:\mathbb{R}^{D_3(d)-1}\to\mathbb{R}^d$, $g(M)=u_k(M)$, is differentiable at $M$.
Writing $G:=Dg(M)\in\mathbb{R}^{d\times (D_3(d)-1)}$ for the Jacobian of $g$ at $M$, we have the asymptotically linear representation
\[
\sqrt{n}\bigl(\hat u_k-u_k\bigr)
\;=\;
G\,\sqrt{n}\,(\hat M-M)\;+\;o_p(1)
\;=\;
\frac{1}{\sqrt{n}}\sum_{i=1}^n \psi_k(X_i)\;+\;o_p(1),
\qquad
\psi_k(x):=G\bigl(M(x)-M\bigr).
\]
Consequently,
\[
\sqrt{n}\bigl(\hat u_k-u_k\bigr)\;\xrightarrow{d}\; \mathcal{N}\!\bigl(0,\;\Sigma_{u_k}\bigr),
\qquad
\Sigma_{u_k}=G\,\Sigma_M\,G^\top,
\quad
\Sigma_M:=\mathrm{Var}\!\big(M(X)\big).
\]
A consistent plug‑in estimator is
\[
\widehat{\Sigma}_{u_k}=\widehat{G}\,\widehat{\Sigma}_M\,\widehat{G}^\top,
\quad
\widehat{\Sigma}_M=\frac{1}{n}\sum_{i=1}^n\bigl(M(X_i)-\hat M\bigr)\bigl(M(X_i)-\hat M\bigr)^\top,
\quad
\widehat{G}=Dg(\hat M).
\]
\end{theorem}

\begin{remark}[Variance estimation and resampling: jackknife]
Although $g$ is complicated, its Jacobian $Dg(M)$ can be obtained numerically or via an eigenvector‑perturbation formula, and the Delta Method then yields a plug‑in variance.
Because $\hat u_k=g(\hat M)$ is a smooth function of \emph{raw} sample averages and is asymptotically linear (Theorem~\ref{thm:main_sample_analogue_theorem}), resampling methods have theoretical guarantees and provide a practical alternative. In particular, the delete‑1 jackknife is convenient for inference on components $[u_k]_j$ or on linear combinations $c^\top u_k$. In all resamples, apply the same normalization and orientation rule as in the main estimate and keep the eigenvalue ordering consistent.
\end{remark}

\section{Simulation Study}

\subsection{Data Generating Process}
To determine the finite sample performance of our method, we design our simulation experiment following the composite structural error setup described in Section~3. This framework captures all three classical sources of endogeneity: measurement error, simultaneity, and omitted variable bias. The data-generating process is given by:
\begin{equation}
\label{eq:data-gen}
\begin{aligned}
    X_1 &= X_1^* + \sqrt{k}\,\epsilon_1,\\
    X_2 &= X_2^* + \sqrt{k}\,\epsilon_2,\\[6pt]
    \begin{bmatrix}
        1  & 1.5\\
        -0.5 & 1
    \end{bmatrix}
    \begin{bmatrix}
        X_1^*\\
        X_2^*
    \end{bmatrix}
    &=
    \begin{bmatrix}
        s_1\\
        s_2
    \end{bmatrix}
    + \sqrt{\tfrac{k}{3}}
    \begin{bmatrix}
        0.5 & -1 & 1.5\\
        -1  & 1  & -1
    \end{bmatrix}
    \begin{bmatrix}
        e_1\\
        e_2\\
        e_3
    \end{bmatrix}.
\end{aligned}
\end{equation}

Here, \(X_1\) and \(X_2\) denote the observed counterparts of the latent 
regressors \(X_1^{*}\) and \(X_2^{*}\), each contaminated by additive 
measurement errors \(\epsilon_1\) and \(\epsilon_2\). 
We treat these errors as \textit{classical}—mean zero and independent of 
all latent economic variables—while allowing them to be strongly 
correlated with one another. 
Specifically, \((\epsilon_1,\epsilon_2)\) is drawn from a bivariate normal distribution 
with unequal variances (marginal variances \(1\) and \(0.25\)) and a high negative covariance 
\(-0.45\) (implying correlation \(-0.9\)). The factor \(\sqrt{k}\) in Eq.~\eqref{eq:data-gen} scales the measurement‑
error variance, so larger \(k\) values \emph{monotonically} worsen the signal‑to‑noise 
ratio and let us trace how estimator performance deteriorates as data 
quality declines. Importantly, our estimator does not require cross‑equation 
independence between the two measurement errors; it remains valid even 
when \(\operatorname{Cov}(\epsilon_1,\epsilon_2)\neq0\).  We therefore 
impose a large negative covariance to stress‑test this robustness. 
A leading empirical case where such correlation occurs—and where the 
relaxation is essential—is the household‑budget data studied by 
\citet{2a001323-07d6-31c1-910a-be0b0a4eecf7}, in which expenditure and quantity are recorded while 
unit price is derived as their ratio.  A positive recording error in 
quantity mechanically induces an equal‑and‑opposite error in the derived 
price, producing the negative correlation we mirror here (with 
\(X_1\) representing \(\log\)‑quantity and \(X_2\) representing 
\(\log\)‑price).

We generate the latent variables \(\{X_1^*, X_2^*\}\) from a
simultaneous‑equations model with structural parameter
\[
\Lambda=\begin{bmatrix}1 & 1.5\\ -0.5 & 1\end{bmatrix},
\]
where the sign pattern replicates a supply‑and‑demand environment (downward‑sloping demand and upward‑sloping supply). The associated structural errors consist of two parts:
\begin{itemize}
    \item \textbf{Skewed equation shifters.}
          \(s_1,s_2 \stackrel{\text{i.i.d.}}{\sim} \mathrm{Gamma}(1,1)\).
          The Gamma family introduces controlled right‑skewness
          (skewness \(=2\), excess kurtosis \(=6\)) while keeping moment
          formulas simple. These latent equation shifters represent exogenous drivers that would qualify as instruments \emph{if observed}; in our benchmarks, IV‑1 uses \(s_2\) directly (oracle), and IV‑2 uses the diluted instrument \(\tilde z=\sqrt{0.3}\,s_2+\sqrt{0.7}\,z\) with \(z\sim\mathcal N(0,1)\).
    \item \textbf{Symmetric omitted effects.}
          The variables \(e_1,e_2,e_3\) are drawn independently from
          Pearson distributions standardized to mean \(0\), variance \(1\),
          and zero skewness, with \emph{non‑excess} kurtosis
          \(\kappa_j \in \{3,4,5\}\) (so \(\kappa{=}3\) corresponds to Normal tails).
          In implementation we use MATLAB’s \texttt{pearsrnd}\((0,1,0,\kappa_j)\),
          which lets us dial tail thickness without affecting mean, variance, or skewness. 
          In this simulation, these are the omitted common shocks entering via the loading matrix in Eq.~\eqref{eq:data-gen}.
\end{itemize}
We assume mutual independence across the three blocks \((s_1,s_2)\), \((e_1,e_2,e_3)\), and \((\epsilon_1,\epsilon_2)\).
The scaling \(\sqrt{k/3}\) in Eq.~\eqref{eq:data-gen} assigns total variance contribution of order \(k\) to the omitted‑shock component across its three independent sources (each contributes \(\approx k/3\)).

Both the measurement errors \(\epsilon\) and the omitted common effects \(\{e_1, e_2, e_3\}\) act as “noise,” inducing correlation in second and higher‑order moments. By varying \(k\), we directly control the magnitude of this noise and investigate its impact on estimation.

\subsection{Mean Squared Error}
Our simulation focuses on recovering the demand elasticity \(-\Lambda_{12}=-1.5\).
For reporting, we equivalently work with the positive slope parameter \(b_1\equiv \Lambda_{12}=1.5\) and compute MSE for \(\hat b_1\) relative to \(b_1\) (IV slopes are negated post‑estimation and signs are aligned so that \(\Lambda_{11},\Lambda_{22}>0\), \(\Lambda_{12}>0\), \(\Lambda_{21}<0\)).

Alongside our proposed method, we implement three other estimators for comparative benchmarking:
\begin{enumerate}
    \item \textbf{Fast‑ICA:} 
    A widely used blind source separation algorithm that relies on non‑Gaussianity and independence of the structural errors. It has also shown very good performance in econometric applications (see \cite{moneta2022identification}). We include Fast‑ICA to illustrate how ignoring correlation among structural errors can affect results. (Implementation: MATLAB \texttt{fastica} with \texttt{'g'='skew'}; initialization is synchronized across conditions via common random numbers.)
    \item \textbf{IV‑1 (``Oracle'' IV):} 
      Two‑stage least squares that treats the latent shifter \(s_{2}\) as an observed
      instrument for \(X_{2}\).
      \emph{Notice that \(s_{2}\) satisfies the two IV conditions:}
      (i) \textit{relevance}—\(X_{2}^{\ast}=s_{2}+\dots\) in the structural system;  
      (ii) \textit{exclusion}—\(s_{2}\) does not appear in the demand
      equation and is independent of the composite error.  
      Since we demean \(y\equiv X_1\), \(x\equiv X_2\), and \(z\) within each replication, 2SLS reduces to \(\hat\beta=(z'y)/(z'x)\); we report \(\hat b_1=-\hat\beta\).
    \item \textbf{IV‑2 (Imperfect IV):} 
      A second 2SLS estimator that uses the diluted instrument
      \(\tilde z = \sqrt{0.3}\,s_{2} + \sqrt{0.7}\,z\) with
      \(z \sim \mathcal N(0,1)\) independent of all model errors.
      The white‑noise contamination weakens relevance while preserving exclusion.
\end{enumerate}

We consider sample sizes \(n\in\{500,3000,5000\}\) and conduct \(10{,}000\) Monte Carlo replicates for each \((n,k)\) and each estimator. Within each replicate we use a common random numbers design: a single “big” dataset is generated and then reused across \((n,k)\); the Fast‑ICA initialization is re‑seeded identically across \((n,k)\) to isolate design effects rather than algorithmic randomness. Point identification for both our proposed method and Fast‑ICA is achieved via sign restrictions.

We measure estimation accuracy using the mean squared error (MSE) between the true parameter \(b_1=1.5\) and its estimate \(\hat b_1\).
Finally, we vary the noise‑to‑signal parameter
\(k\in\{0,0.1,\dots,0.5\}\).
Increasing \(k\) inflates the measurement‑error component
\((\sqrt{k}\,\epsilon)\) and the omitted‑shock component
\((\sqrt{k/3}\,e)\), thereby increasing the variance of the composite error.
Because both variances and covariances scale proportionally in \(k\), the \emph{cross‑equation correlation} of the composite error remains unchanged for any \(k>0\) (it is undefined at \(k=0\)).

 From Table~1, we see that Fast‑ICA (F‑ICA) performs very well at \(k=0\), the independence case: its MSE is comparable to (and for larger \(n\), slightly below) that of the proposed eigenvector estimator (M1). Once \(k>0\), however, F‑ICA’s MSE grows sharply with \(k\) and, at higher noise levels, even fails to improve with larger \(n\) (e.g., at \(k=0.5\), MSE rises from \(0.916\) at \(n=500\) to \(1.249\) at \(n=5000\)), indicating asymptotic inconsistency under dependent composite shocks. The mechanism is subtle: although the fixed‑point iteration itself optimizes a higher‑moment criterion and does not use second moments directly, identification in Fast‑ICA relies on a \emph{whitening} step that premultiplies the data by the inverse square root of the covariance, yielding \(z=\Sigma_x^{-1/2}x=B s\). This step forces the subsequent search over an \emph{orthogonal} demixing matrix, which is valid only when the latent signals (here, the composite structural errors) are uncorrelated —a condition implied by independence in ICA, but violated in our design for \(k>0\). Consequently, no orthogonal demixer can recover the sources, and the algorithm converges to a biased limit even though its higher‑order objective is correctly specified within the (incorrect) orthogonal constraint set. This failure mode is representative: many higher‑cumulant‑based procedures impose whitening‑plus‑orthogonality as a precondition, and thus inherit a similar bias in this setup.

By contrast, M1’s MSE increases only modestly with \(k\) but declines sharply as \(n\) grows (e.g., from \(0.0113\) to \(0.00123\) at \(k=0\), and from \(0.0524\) to \(0.00439\) at \(k=0.5\)), consistent with the root‑\(n\) rate predicted by our asymptotic theory. Quantitatively, for \(k>0\) F‑ICA is about \(4\times\) to \(284\times\) less accurate than M1 across our \((n,k)\) grid.

Turning to the IV benchmarks, IV‑1 (oracle) uniformly dominates, as expected from using the latent shifter \(s_2\): its MSE is about half of M1’s at \(k=0\) and about one‑fifth to one‑quarter at \(k=0.5\) across \(n\). Importantly, although M1 does not use any instrument, it achieves the \emph{same order of magnitude} MSE as IV‑2, which benefits from partial access to the latent shifter through \(\tilde z=\sqrt{0.3}\,s_2+\sqrt{0.7}\,z\). Across all \((n,k)\), the MSE ratio \(\text{M1}/\text{IV-2}\) lies between \(\approx 0.53\) and \(\approx 1.43\) (median \(\approx 0.83\)), with M1 typically outperforming IV‑2 for low‑to‑moderate noise (\(k\le 0.3\); e.g., \(n=5000,k=0.2\): \(0.00201\) vs \(0.00264\)) and IV‑2 modestly ahead at higher noise (\(k\ge 0.4\); e.g., \(n=5000,k=0.5\): \(0.00439\) vs \(0.00350\)).

Overall, this experiment shows that neglecting dependence in the composite errors is highly consequential for whitening‑based, higher‑order methods, whereas M1 remains reliable across finite samples and a wide range of noise levels—delivering near‑IV performance without any instrument.

\begin{table}[H]
\centering
\scriptsize
\caption{Comparison of Methods at Various Sample Sizes (Vertical Split)}
\label{tab:vertical}

\subcaption*{\textbf{$n = 500$}}
\begin{tabular}{l|cccc}
\toprule
\multicolumn{1}{c|}{$k$}
 & M1 & F-ica & IV-1 & IV-2 \\
\midrule
0 & 1.13$\times 10^{-2}$ & 1.12$\times 10^{-2}$ & 6.18$\times 10^{-3}$ & 2.15$\times 10^{-2}$ \\
0.1 & 1.52$\times 10^{-2}$ & 6.15$\times 10^{-2}$ & 7.05$\times 10^{-3}$ & 2.45$\times 10^{-2}$ \\
0.2 & 2.07$\times 10^{-2}$ & 0.204                 & 7.91$\times 10^{-3}$ & 2.75$\times 10^{-2}$ \\
0.3 & 2.81$\times 10^{-2}$ & 0.430                 & 8.77$\times 10^{-3}$ & 3.05$\times 10^{-2}$ \\
0.4 & 3.85$\times 10^{-2}$ & 0.686                 & 9.63$\times 10^{-3}$ & 3.35$\times 10^{-2}$ \\
0.5 & 5.24$\times 10^{-2}$ & 0.916                 & 1.05$\times 10^{-2}$ & 3.66$\times 10^{-2}$ \\
\bottomrule
\end{tabular}

\vspace{1em}

\subcaption*{\textbf{$n = 3000$}}
\begin{tabular}{l|cccc}
\toprule
\multicolumn{1}{c|}{$k$}
 & M1 & F-ica & IV-1 & IV-2 \\
\midrule
0 & 1.97$\times 10^{-3}$ & 1.67$\times 10^{-3}$ & 1.01$\times 10^{-3}$ & 3.42$\times 10^{-3}$ \\
0.1 & 2.53$\times 10^{-3}$ & 3.08$\times 10^{-2}$ & 1.15$\times 10^{-3}$ & 3.89$\times 10^{-3}$ \\
0.2 & 3.29$\times 10^{-3}$ & 0.140                 & 1.29$\times 10^{-3}$ & 4.37$\times 10^{-3}$ \\
0.3 & 4.30$\times 10^{-3}$ & 0.409                 & 1.43$\times 10^{-3}$ & 4.84$\times 10^{-3}$ \\
0.4 & 5.63$\times 10^{-3}$ & 0.806                 & 1.57$\times 10^{-3}$ & 5.31$\times 10^{-3}$ \\
0.5 & 7.32$\times 10^{-3}$ & 1.22                  & 1.71$\times 10^{-3}$ & 5.78$\times 10^{-3}$ \\
\bottomrule
\end{tabular}

\vspace{1em}

\subcaption*{\textbf{$n = 5000$}}
\begin{tabular}{l|cccc}
\toprule
\multicolumn{1}{c|}{$k$}
 & M1 & F-ica & IV-1 & IV-2 \\
\midrule
0 & 1.23$\times 10^{-3}$ & 9.31$\times 10^{-4}$ & 6.23$\times 10^{-4}$ & 2.05$\times 10^{-3}$ \\
0.1 & 1.56$\times 10^{-3}$ & 2.87$\times 10^{-2}$ & 7.02$\times 10^{-4}$ & 2.35$\times 10^{-3}$ \\
0.2 & 2.01$\times 10^{-3}$ & 0.132                 & 7.84$\times 10^{-4}$ & 2.64$\times 10^{-3}$ \\
0.3 & 2.61$\times 10^{-3}$ & 0.402                 & 8.66$\times 10^{-4}$ & 2.93$\times 10^{-3}$ \\
0.4 & 3.39$\times 10^{-3}$ & 0.815                 & 9.48$\times 10^{-4}$ & 3.22$\times 10^{-3}$ \\
0.5 & 4.39$\times 10^{-3}$ & 1.25                  & 1.03$\times 10^{-3}$ & 3.50$\times 10^{-3}$ \\
\bottomrule
\end{tabular}

\vspace{1em}

\end{table}

\subsection{Inference}
Table~2 reports empirical coverage of two-sided 95\% confidence intervals for the demand elasticity \(b_1=\Lambda_{12}=1.5\) (equivalently, \(-\Lambda_{12}\)). We use the same DGP as in the simulation study with the noise level fixed at \(k=0.5\), across sample sizes \(n\in\{500,3000,5000\}\) and 5{,}000 Monte Carlo replications.

As discussed in Section~6, the variances are estimated using two methods: the leave‑one‑out jackknife and the delta method. Confidence intervals are then constructed by asymptotic normal approximation using the standard normal critical value \(z_{0.975}\).

As shown in Table~2, both procedures achieve coverage close to the nominal 95\% level. The jackknife is slightly closer to nominal at \(n=500\), while for \(n\ge 3000\) the two methods are essentially indistinguishable. With 5{,}000 replications, the Monte Carlo standard error of a 95\% coverage estimate is about 0.003 (i.e., roughly \(\pm 0.6\) percentage points for a 95\% simulation band), so small differences of that magnitude are not substantively meaningful. These findings are consistent with our asymptotic normality derivation and support the use of both plug‑in (delta) and resampling (jackknife) inference in this setting.

\begin{table}[ht!]
\centering
\caption{Coverage rates (\%) for 95\% confidence intervals, $k = 0.5$}
\label{tab:coverage}
\begin{tabular}{lccc}
\toprule
\textbf{Method} & \textbf{$n=500$} & \textbf{$n=3000$} & \textbf{$n=5000$} \\
\midrule
Jackknife  & 94.3 & 95.3 & 95.5 \\
Delta method                       & 90.7 & 94.4 & 94.3 \\
\bottomrule
\end{tabular}
\end{table}

\subsection{Over-identification test: size and power}

We study the finite-sample behavior of the over-identification test from Section~5 (Extension~3) using the \emph{same} data-generating process and simulation protocol as in Section~7.1. In each replication we estimate the (scaled/permuted) structural parameter matrix \(\widehat{\widetilde{\Lambda}}\) \emph{solely} from third-cumulant information via the eigenvector method, fixing \(w_2=\mathbf{1}\) and drawing a single \(w_1\) at random (uniformly from the unit cube, as in Theorem~5.1) once and for all. We then test the null of the joint validity of uncorrelated structural errors and A1$^\prime$–A3$^\prime$ by examining whether \(\widehat{\Theta}=\widehat{\widetilde{\Lambda}}\;\widehat{\Var}(X)\;\widehat{\widetilde{\Lambda}}^{\!\top}\) is diagonal. With \(d=2\) this yields a single over-identifying restriction; we use a Wald statistic with variance computed by \emph{delta‑method} linearization and compare to a \(\chi^2_1\) critical value at \(\alpha=5\%\).

Table~\ref{tab:OID_power_short} reports rejection rates for sample sizes \(n\in\{500,750,1000,5000\}\) and noise scales \(k\in\{0,0.1,\dots,0.5\}\) (larger \(k\) produces stronger departures from joint diagonality). Under the null (\(k=0\)) the test modestly over-rejects at small samples (10.4\% at \(n=500\)) but improves with \(n\) (6.0\% by \(n=5000\)), indicating that the size distortion recedes as the sample grows. Power rises with both the severity of the violation and with \(n\): for instance, at \(k=0.2\) the rejection rate increases from 0.442 at \(n=500\) to 0.813 at \(n=1000\); by \(n=5000\) power is essentially one for all \(k\ge 0.2\) (and 0.985 even at \(k=0.1\)). At very low sample sizes the test still struggles to detect mild violations—performance at \(n=500\) remains moderate—so some VAR applications with monthly data may face limited power against small departures from joint diagonality. That said, the variance regime used here is deliberately challenging; in many empirical settings (e.g., when working with residuals after partialling out lags) the effective noise level is lower, which should improve performance.

\begin{table}[H]
\centering
\caption{Rejection rates (rows: \(n\); columns: \(k\)). Nominal size \(5\%\).}
\label{tab:OID_power_short}
\begin{tabular}{lcccccc}
\toprule
& \multicolumn{6}{c}{Noise scale \(k\)}\\
\cmidrule(lr){2-7}
\(n\) & 0.0 & 0.1 & 0.2 & 0.3 & 0.4 & 0.5\\
\midrule
500  & 0.104 & 0.161 & 0.442 & 0.615 & 0.683 & 0.691 \\
750  & 0.087 & 0.243 & 0.654 & 0.840 & 0.905 & 0.917 \\
1000 & 0.083 & 0.346 & 0.813 & 0.944 & 0.973 & 0.980 \\
5000 & 0.060 & 0.985 & 1.000 & 1.000 & 1.000 & 1.000 \\
\bottomrule
\end{tabular}
\end{table}

\section{Application}

\subsection{Return to Education}

Our first application estimates \emph{returns to schooling} with the proposed method. Data come from the Wooldridge textbook empirical exercise; details are provided in the \texttt{wooldridge} package manual on CRAN.\footnote{\url{https://cran.r-project.org/web/packages/wooldridge/wooldridge.pdf}}

Following \citet{card1993using}, we adopt the linear system
\begin{equation}
\begin{aligned}
lwage_i &= X_i \alpha + educ_i^{obs} \beta + e_{1i},\\
educ_i^{obs} &= X_i \gamma + e_{2i}.
\end{aligned}
\end{equation}
Here, \(lwage_i\) denotes the natural log of wages, \(educ_i^{obs}\) is reported schooling, and \(X_i\) denotes a vector of control variables—such as race, potential experience, and location dummies. Following \citet{card1993using}, there are two main sources of inconsistency in the OLS estimate of \(\beta\). The first is measurement error in reported schooling. Consistent with the original study, we model this error as \emph{classical}: mean‑zero, independent of the true education level and other regressors, and approximately normal (specifically, symmetric). Let \(educ_i^*\) denote true schooling and write the measurement equation
\[
educ_i^{obs} = educ_i^* + \Delta_{educ,i}.
\]

The second—and more challenging—issue is omitted variable bias. As emphasized by \citet{card1993using}, much of the literature attributes this bias to unobserved “ability,” a latent individual trait that is inherently difficult to measure. In this paper, we assume that ability is the primary source of structural error correlation. That is, if ability were observed, all parameters in the model could be estimated consistently. Moreover, we assume that ability, conditional on \(X_i\), is symmetrically distributed. This assumption is consistent with precedents in applied econometrics (e.g., \citet{ee98fe21-acdf-39c2-ada6-7518abe3ad0b},\citet{doi:10.1086/504455} ) and is supported by empirical findings from psychometric studies (e.g., \cite{jensen1998g}). 

In summary, we use the following triangular linear structural model, which mirrors the composite structural error framework discussed in Section 3 and verified by simulation in Section 7:
\[
\begin{aligned}
lwage_i &= X_i\alpha + educ_i^{obs}\beta + e_{1i},\\
educ_i^{obs} &= X_i\gamma + e_{2i},\\
e_{1i} &= ability_i\theta_1 - \Delta_{educ,i}\,\beta + u_i,\\
e_{2i} &= ability_i\theta_2 + \Delta_{educ,i} + v_i,
\end{aligned}
\]
where \(u_i\) and \(v_i\) are the structural errors after accounting for \(ability_i\). These errors capture factors with no cross‑equation effects independent of ability and measurement error. The term $v_i$ collects idiosyncratic shocks that affect schooling but not wages 
except through schooling (given $X_i$); in IV setups these shocks are the source of 
exogenous variation typically targeted by instruments. We assume \(u_i\) and \(v_i\) are mutually independent and have non-zero skewness. \(\Delta_{educ,i}\) denotes measurement error in education, assumed to be normally distributed and independent of the true level of education. We also assume \(ability_i\) is the sole source of omitted‑variable bias and is independent of \(u_i\) and \(v_i\) (conditional on \(X_i\)). Finally, once the effects of \(X_i\) are accounted for, the distribution of \(ability_i\) is assumed symmetric.

As shown in Table 2 of \citet{card2001estimating}, previous estimates of the return to schooling range from $0.0245$ to $0.36$. \cite{card1993using} IV estimate based on distance to college is $0.132$. Using our proposed estimator, we obtain a point estimate of $0.0987$ with a 95\% jackknife confidence interval of \([0.0358,\,0.1500]\). This interval contains the majority of the estimates reviewed by \cite{card2001estimating} in this setup, and our point estimate is close to the distance‑to‑college IV benchmark, lending credibility to our approach.

\subsection{Uncertainty and Business Cycle}

Identification through higher moments or cumulant restrictions has recently become more popular in the macroeconometrics literature, particularly within the vector autoregression (VAR) framework. As highlighted in the introduction, in this literature the assumption of uncorrelated structural shocks is not only relatively uncontroversial but also often desirable for interpretable impulse responses. Our results show that the commonly imposed model—characterized by a diagonal second cumulant and diagonal third or fourth cumulants of the structural shocks—is \emph{overidentified} and therefore testable. One reason structural shocks can appear dependent in practice is that the specified VAR system is simply too small, so the included variables do not support a linear causal model. The empirical exercise below illustrates how the test proposed in Section~5 can detect such misspecification and how the resulting evidence can answer a substantive economic question. 

Uncertainty typically rises in economic downturns. \citet{ludvigson2021uncertainty} use a VAR to ask whether uncertainty helps cause recessions or is instead an endogenous response. A central conclusion of the paper is that the \emph{type} of uncertainty matters: innovations to financial uncertainty (\(UF\)) behave more like an exogenous driver, raising macro uncertainty (\(MU\)) and lowering industrial production (\(IP\)), whereas downturns in \(IP\) have limited feedback to \(UF\).\footnote{See \citet{ludvigson2021uncertainty} for details on the uncertainty measures and VAR specification.} In their analysis, different variables were considered to capture real (macro) uncertainty. The main specification used \(MU\) to obtain the aforementioned result, while the economic policy uncertainty (EPU) index constructed by \citet{10.1093/qje/qjw024} was employed as a robustness check. Interestingly, in that specification \(UF\) exhibited a statistically significant contemporaneous response to \(IP\). However, the EPU sample there was much shorter (358 observations).

Importantly, the view that \(UF\) acts as a driver of the business cycle is not isolated. \citet{https://doi.org/10.1002/jae.2672} use a heteroskedasticity‑based identification method and conclude that financial uncertainty does not respond to shocks in real activity nor to shocks in macro uncertainty. \citet{DAVIS2023180} revisit the same question using an ICA‑based approach: they posit that the reduced‑form residuals are linear combinations of three independent non‑Gaussian shocks and develop a permutation‑based independence test under a proposed causal ordering. Within the \((MU,UF,IP)\) system, the lower‑triangular orderings that place \(UF\) on top (most exogenous) are not rejected, providing further support for \citet{ludvigson2021uncertainty}. When \(MU\) is replaced by EPU, however, independence is rejected for all triangular orderings. This may reflect data limitations (the shorter EPU sample used in those papers), or it may indicate that a strict triangular contemporaneous structure does not hold (EPU also responds contemporaneously to \(IP\)). Building on the statistical framework of \citet{DAVIS2023180}, we therefore re‑estimate the EPU specification using updated EPU data that cover the full \citet{ludvigson2021uncertainty} sample, 1960:07–2015:04. Our goal is to assess the claim that \(UF\) is approximately exogenous by evaluating two conditions: (1) shocks to \(UF\) act as the exogenous driver; and (2) the shocks are uncorrelated and satisfy A2$^\prime$–A3$^\prime$.

We adopt \citet{DAVIS2023180}'s VAR set‑up
\[
Y_t \;=\; A_1 Y_{t-1} + \cdots + A_p Y_{t-p} + e_t,
\qquad
e_t \;=\; B u_t,
\]
with \(Y_t'=(UF_t,EPU_t,IP_t)\) and \(u_t\) a three‑dimensional vector of shocks. We assume \(u_t\) are i.i.d.\ across \(t\), and we set the lag length to \(p=6\) as in \citet{ludvigson2021uncertainty}. The structural object of interest is \(B\). Our test, however, differs from \citet{DAVIS2023180} (and, e.g., \citet{10.1257/pandp.20221047}): we test only the \emph{minimal} restriction set that suffices for identification here—diagonality of a single higher‑order cumulant (third order) together with uncorrelated structural shocks. We implement the joint‑diagonality over‑identification test on the estimated reduced‑form residuals. We assume the lag order is correctly specified so that the Delta‑method large‑sample approximation for our Wald statistic (Section~5) is valid.

Ideally, one would implement a joint three‑equation test and, if not rejected, recover the full \(B\) from third‑order information and then inspect any implied lower‑triangular pattern. Given that the sample size in this exercise is still relatively moderate to estimate higher cumulants accurately, we instead focus on three testable implications that follow if \(UF\) is effectively “on top” of the contemporaneous causal order. The idea of the test mirrors \citet{DAVIS2023180}, which effectively tests the joint plausibility of a given rotation combined with independent shocks. 

Concretely, if the first row of \(B\) loads only on \(u_{UF}\) (we do not require full triangularity beyond this) and the shocks \(u_t\) satisfy A2$^\prime$–A3$^\prime$ \emph{and} uncorrelatedness, then the \(2\times2\) subsystems \((e_{UF},e_{EPU})\) and \((e_{UF},e_{IP})\) can be written as linear combinations of shocks that satisfy joint diagonality of the third and second cumulants. To see this,
\[
\begin{bmatrix}
e_{UF}\\[2pt]
e_{EPU}\\[2pt]
e_{IP}
\end{bmatrix}
=
\begin{bmatrix}
1 & 0 & 0\\
b_{21} & 1 & b_{23}\\
b_{31} & b_{32} & 1
\end{bmatrix}
\begin{bmatrix}
u_{UF}\\[2pt]
u_{EPU}\\[2pt]
u_{IP}
\end{bmatrix},
\]
so \(e_{UF}=u_{UF}\), \(e_{EPU}=b_{21}u_{UF}+u_{EPU}+b_{23}u_{IP}\), and \(e_{IP}=b_{31}u_{UF}+b_{32}u_{EPU}+u_{IP}\). Let \(u_{\mathrm{comb}1}:=u_{EPU}+b_{23}u_{IP}\) and \(u_{\mathrm{comb}2}:=b_{32}u_{EPU}+u_{IP}\). Then
\[
T_{UF,EPU}
\begin{bmatrix}
e_{UF}\\[2pt]
e_{EPU}
\end{bmatrix}
=
\begin{bmatrix}
u_{UF}\\[2pt]
u_{\mathrm{comb}1}
\end{bmatrix},
\qquad
T_{UF,EPU} :=
\begin{bmatrix} 1 & 0 \\ -\,b_{21} & 1 \end{bmatrix},
\]
and analogously
\(
T_{UF,IP}\begin{bmatrix}e_{UF}\\ e_{IP}\end{bmatrix}
=
\begin{bmatrix}u_{UF}\\ u_{\mathrm{comb}2}\end{bmatrix}
\)
with
\(
T_{UF,IP}=\begin{bmatrix}1&0\\ -\,b_{31}&1\end{bmatrix}.
\)
If \(u_t\) satisfies A2$^\prime$–A3$^\prime$ and uncorrelatedness, the pairs \((u_{UF},u_{\mathrm{comb}1})\) and \((u_{UF},u_{\mathrm{comb}2})\) retain diagonal third cumulants and diagonal covariance, so each \(2\times2\) subsystem meets the joint diagonality restriction under the null.

By contrast, the subsystem \((e_{EPU},e_{IP})\) generically mixes three shocks in a two‑dimensional space. Absent knife‑edge cancellations (e.g., \(b_{21}=b_{31}=0\) or vanishing third cumulants in just the right combination), \((e_{EPU},e_{IP})\) will violate joint diagonality of the second and third cumulants. This is why we can test the joint validity of the assumptions by focusing on these three two‑equation tests.

Applying this procedure, the over‑identification test rejects joint diagonality for \((e_{EPU},e_{IP})\) at the 5\% level (Delta‑method \(p<0.01\)), while it does not reject for the pairs \((e_{UF},e_{EPU})\) and \((e_{UF},e_{IP})\) at the 5\% level. On the surface, this pattern is what the “\(UF\) on top’’ hypothesis predicts: pairs that include \(UF\) behave like two‑source systems and pass, whereas \((EPU,IP)\) fails the minimal over‑identifying restriction. Looking more closely, however, \((e_{UF},e_{IP})\) is \emph{borderline} at the 5\% level (Delta‑method \(p=0.055\)), which is suggestive of modest tension with exact joint diagonality—consistent with a small contemporaneous channel between uncertainty and policy uncertainty or with specification noise. By contrast, the large \(p\)‑value for \((e_{UF},e_{EPU})\) (\(p=0.65\)) indicates no detectable deviation from the two‑source benchmark in that pair. Taken together, these results are broadly in line with \citet{https://doi.org/10.1002/jae.2672}, in the sense that uncertainty shocks appear to help drive the cycle, while also indicating that \(UF\) is not perfectly exogenous in the EPU specification. Moreover, it shows that a strict triangular contemporaneous ordering may be too rigid for this model, as \citet{ludvigson2021uncertainty} suggested. 

\section{Acknowledgments}

I would like to express my sincere gratitude to my supervisors, Aureo de Paula and Andrei Zeleneev, for their continuous support, guidance, and encouragement throughout this work. I am also indebted to Tim Christensen, Ben Deaner, Raffaella Giacomini, Dennis Kristensen, and Daniel Lewis for their invaluable feedback and advice. I am grateful to Yanziyi Zhang and Chen-Wei Hsiang for many helpful discussions. Any remaining errors are solely my own responsibility.

\section{Declaration of competing interest}
The author declares that they have no known competing financial interests or personal relationships that could have appeared to influence the work reported in this paper.

\appendix
\noindent{\textbf{\Large{Appendix}}}

    \section{Theorem 5.1 for general \(\boldsymbol{h>2}\) }
     \begin{proof}
Let $v=A^\top w$ so that, by additivity and homogeneity of cumulants,
\[
Q(w)=\kappa_h(w^\top X)=\kappa_h(v^\top S)=\sum_{i=1}^d \kappa_h(S_i)\,v_i^{\,h}.
\]
Differentiating w.r.t.\ $v$ gives
\[
\frac{\partial^2 Q}{\partial v_i\partial v_j}
= h(h-1)\,\kappa_h(S_i)\,v_i^{\,h-2}\,\mathbf{1}\{i=j\},
\]
so $\nabla_v^2 Q(v)=\mathrm{diag}\!\big(h(h-1)\kappa_h(S_i)\,v_i^{\,h-2}\big)$ and, by the chain rule,
\[
\nabla_w^2 Q(w)=A\,\nabla_v^2 Q(v)\,A^\top \;=\; A\,D_w\,A^\top, 
\quad D_w:=\mathrm{diag}\!\big(h(h-1)\kappa_h(S_i)\,v_i^{\,h-2}\big).
\]

Fix $w_2$ such that $v_{2,i}=(A^\top w_2)_i\neq 0$ for all $i$ (e.g., $w_2=\mathbf{1}$ with $\sum_q A_{qi}\neq 0$ for each $i$). Then
\[
H \;:=\; \big(\nabla_w^2 Q(w_2)\big)^{-1}\,\nabla_w^2 Q(w_1)
= (A^\top)^{-1}\,D_{w_2}^{-1}\,D_{w_1}\,A^\top
= \Lambda^\top \,\mathrm{diag}(\lambda_i)\,(\Lambda^\top)^{-1},
\]
where
\[
\lambda_i = \frac{h(h-1)\kappa_h(S_i)\,v_{1,i}^{\,h-2}}{h(h-1)\kappa_h(S_i)\,v_{2,i}^{\,h-2}}
= \Big(\frac{v_{1,i}}{v_{2,i}}\Big)^{h-2}, 
\qquad v_{\ell,i}=(A^\top w_\ell)_i \;( \ell=1,2).
\]
Hence the eigenvectors of $H$ are the columns of $\Lambda^\top$ (equivalently, the rows of $\Lambda$).

It remains to show that, for $w_1$ drawn from any absolutely continuous distribution (e.g., uniform on $[0,1]^d$), the eigenvalues $\{\lambda_i\}$ are almost surely distinct. Set $k:=h-2\ge 1$ and fix $i\neq j$.
Equality $\lambda_i=\lambda_j$ is equivalent to
\[
\Big(\frac{v_{1,i}}{v_{2,i}}\Big)^{k}=\Big(\frac{v_{1,j}}{v_{2,j}}\Big)^{k}
\;\;\Longleftrightarrow\;\;
v_{2,j}^{k}\,\big(w_1^\top A_{\cdot i}\big)^{k} \;=\; v_{2,i}^{k}\,\big(w_1^\top A_{\cdot j}\big)^{k}.
\]

\paragraph{Case 1: $k$ odd.}
Taking $k$-th roots preserves sign:
\[
v_{2,j}\,(w_1^\top A_{\cdot i}) \;=\; v_{2,i}\,(w_1^\top A_{\cdot j})
\;\;\Longleftrightarrow\;\;
w_1^\top\big(v_{2,j}A_{\cdot i}-v_{2,i}A_{\cdot j}\big)=0.
\]
Because $A_{\cdot i}$ and $A_{\cdot j}$ are linearly independent and $v_{2,i},v_{2,j}\neq 0$, the vector
$v_{2,j}A_{\cdot i}-v_{2,i}A_{\cdot j}\neq 0$, so this is a $(d-1)$-dimensional hyperplane (Lebesgue measure zero).

\paragraph{Case 2: $k$ even.}
Equality of $k$-th powers is equality of absolute values:
\[
\big|v_{2,j}\,(w_1^\top A_{\cdot i})\big| \;=\; \big|v_{2,i}\,(w_1^\top A_{\cdot j})\big|
\;\;\Longleftrightarrow\;\;
\big(v_{2,j}\,w_1^\top A_{\cdot i}\big)^2 = \big(v_{2,i}\,w_1^\top A_{\cdot j}\big)^2.
\]
Factoring gives
\[
\big[w_1^\top\big(v_{2,j}A_{\cdot i}-v_{2,i}A_{\cdot j}\big)\big]\;
\big[w_1^\top\big(v_{2,j}A_{\cdot i}+v_{2,i}A_{\cdot j}\big)\big]\;=\;0,
\]
so $\lambda_i=\lambda_j$ iff $w_1$ lies in the \emph{union} of the two hyperplanes
$w_1^\top(v_{2,j}A_{\cdot i}-v_{2,i}A_{\cdot j})=0$ or
$w_1^\top(v_{2,j}A_{\cdot i}+v_{2,i}A_{\cdot j})=0$.
Each is a $(d-1)$-dimensional hyperplane (nonzero normal because $A_{\cdot i},A_{\cdot j}$ are independent and $v_{2,i},v_{2,j}\neq 0$), hence measure zero; a finite union is also measure zero.

\medskip
Taking a finite union over all $i\neq j$, the event that any two eigenvalues coincide has Lebesgue measure zero. Therefore, for $w_1$ drawn from an absolutely continuous distribution on $\mathbb{R}^d$ (e.g., uniform on the unit cube), the eigenvalues are almost surely pairwise distinct, and the corresponding eigendirections are unique up to scaling and permutation. This completes the proof.
\end{proof}

\subsection{Estimation for Extension 1 and 2 (tall $A$, $d_1>d_2$)}

When $A\in\mathbb{R}^{d_1\times d_2}$ is full column rank and A2$^\prime$–A3$^\prime$ hold (skewed factors), identification uses
\[
H \;=\; G(w_1)\,G(w_2)^{+}
\;=\; A\,\big(D_{w_1}D_{w_2}^{-1}\big)\,A^{+}
\;=\; A\,\Delta\,A^{+},
\qquad
\Delta=\operatorname{diag}\!\Big(\tfrac{(A^\top w_1)_i}{(A^\top w_2)_i}\Big)_{i=1}^{d_2},
\]
so the $d_2$ \emph{right} eigenvectors of $H$ corresponding to its nonzero eigenvalues are precisely the columns of $A$ (up to column scaling and permutation).

\medskip
\noindent\textbf{Implementation.}
Compute $\hat G(w_\ell)=\nabla_w^2\hat Q(w_\ell)$ for $\ell=1,2$ from centered data (equivalently, apply the raw–to–centered polynomial map inside the construction).
Choose $w_2$ so that $(A^\top w_2)_i\neq 0$ for all relevant $i$ (e.g., $w_2=\mathbf{1}$ with nonzero column sums).
Stabilize $\hat G(w_2)^{+}$ via SVD (i.e, take the best rank \(d_2\) approximate) and form
\[
\hat H \;=\; \hat G(w_1)\,\hat G(w_2)^{+}\in\mathbb{R}^{d_1\times d_1}.
\]
Let $\{\hat v_j\}$ be the \emph{right} eigenvectors associated with the largest nonzero eigenvalues of $\hat H$ (if $d_2$ is known, take exactly $d_2$ of them; otherwise select by a nonzero–spectrum threshold which shrinks slower than \(\sqrt{n}\)).
Then $\{\hat v_j\}$ estimate the columns of $A$ up to scaling and permutation; apply the same column orientation / normalization rule as in the square case.
Inference follows by the same Delta–method arguments as in Section~5, since the stabilized pseudoinverse is smooth on a rank–constant neighborhood and eigenvalue separation holds almost surely for random $w_1$.

\section{Overidentification test under uncorrelated structural errors}

This subsection derives the asymptotic $\chi^2$ distribution of the overidentification test described in Section~4 (Extension~3). Throughout, we assume A1$^\prime$--A3$^\prime$, the eigenvalue distinctness condition as in Theorem~\ref{thm:eig-decomp} (e.g. satisfied almost surely when \ $w_2=\mathbf{1}$ with $\sum_q A_{qi}\neq 0$ for all $i$ and $w_1$ drawn from an absolutely continuous distribution), and the same eigenvector orientation convention (row‑sum or unique largest‑entry) used in Section~5.

\paragraph{Setup and notation.}
Recall that $X_c:=X-\mathbb{E}[X]$, $\Sigma:=\mathrm{Var}(X)=\mathbb{E}[X_cX_c^\top]$. Let
\[
H(M)\;=\;\bigl(\nabla_w^2 Q(w_2)\bigr)^{-1}\,\nabla_w^2 Q(w_1)\,,
\quad
Q(w)=\kappa_3(w^\top X)=\mathbb{E}\!\big[(w^\top X_c)^3\big],
\]
where $M=\mathbb{E}[M(X)]$ stacks the \emph{raw} monomials of total degree $1$–$3$ (cf.\ Section~5). Denote by
\[
\widetilde{\Lambda}(M)\in\mathbb{R}^{d\times d}
\]
the \emph{scaled and permuted} demixing matrix recovered solely from third‑cumulant information, i.e.\ the matrix whose rows are the (oriented, normalized) eigenvectors of $H(M)$. The orientation is fixed by the same rule as in Section~5, ensuring differentiability at the population $M$.

Let $\vech_{\mathrm{off}}$ denote the operator that stacks the strict upper‑triangular entries of a symmetric matrix into a vector in $\mathbb{R}^{\binom{d}{2}}$.

\paragraph{Population restrictions and sample counterpart.}
Under uncorrelated structural errors, $\Sigma=A D_2 A^\top$ with diagonal $D_2$, hence
\[
\widetilde{\Lambda}(M)\,\Sigma\,\widetilde{\Lambda}(M)^\top
= P^* D^* \Lambda\,\Sigma\,\Lambda^\top D^* P^{*\,\top}
= P^* D^* D_2 D^* P^{*\,\top}
\quad\text{is diagonal.}
\]
Equivalently,
\[
r_0 \;:=\; \vech_{\mathrm{off}}\!\Big(\,\widetilde{\Lambda}(M)\,\Sigma\,\widetilde{\Lambda}(M)^\top\,\Big) \;=\; 0 \in \mathbb{R}^{\binom{d}{2}}.
\]

In samples, define the raw‑moment average $\hat M=\tfrac{1}{n}\sum_{i=1}^n M(X_i)$ and the centered sample covariance
\[
\widehat{\Sigma}_X=\frac{1}{n}\sum_{i=1}^n (X_i-\bar X)(X_i-\bar X)^\top,\qquad \bar X=\tfrac1n\sum_{i=1}^n X_i.
\]
Construct $\hat{\widetilde{\Lambda}}=\widetilde{\Lambda}(\hat M)$ from third cumulants only (as in Section~5), and set
\[
\hat r \;=\; \vech_{\mathrm{off}}\!\Big(\,\hat{\widetilde{\Lambda}}\,\widehat{\Sigma}_X\,\hat{\widetilde{\Lambda}}^\top\,\Big)\in\mathbb{R}^{\binom{d}{2}}.
\]

\paragraph{Joint moment vector and smooth map.}
Let $Z(X)\in\mathbb{R}^{p}$ stack all raw monomials of total degree $1$–$3$ (the vector used to build $\hat M$). Notice, this includes the degree‑$1$ and degree‑$2$ raw monomials needed to build $\Sigma$ from raw moments (i.e.\ means and second moments). Then $\theta:=\mathbb{E}[Z(X)]$ collects all raw moments required to compute both $\widetilde{\Lambda}$ and $\Sigma$; write $\hat\theta=\tfrac1n\sum_{i=1}^n Z(X_i)$.

Define the composition
\[
R(\theta)
\;:=\;
\vech_{\mathrm{off}}\!\Big(\,\widetilde{\Lambda}\big(C_3(\theta)\big)\;\Sigma(\theta)\;\widetilde{\Lambda}\big(C_3(\theta)\big)^\top\,\Big),
\]
where $C_3(\theta)$ is the polynomial map from raw moments (components of $\theta$) to centered third moments (the third cumulant) (Section~5), $\Sigma(\theta)$ is the polynomial map from raw moments to the covariance $\Sigma$, and $\widetilde{\Lambda}$ is the (oriented) eigenvector map applied to $H\big(C_3(\theta)\big)$. Then $r_0=R(\theta)$ and $\hat r=R(\hat\theta)$.

\begin{assumption}[Sampling, separation, and nonredundancy]\label{ass:appendix_joint}
The following hold:
\begin{enumerate}[(i)]
    \item \textbf{Sampling and moments.} The sample $\{X_i\}_{i=1}^n$ is i.i.d.\ with $\mathbb{E}\|X\|^{6}<\infty$.
    \item \textbf{Eigenvalue separation and orientation.} The eigenvalues of $H\!\big(C_3(\theta)\big)$ are simple and the eigenvectors are oriented by the same rule as in Section~5 (e.g., row‑sum $>0$ or unique largest entry), so that the eigenvector map is differentiable at the population. Concretely, as in Theorem~\ref{thm:eig-decomp}, take $w_2=\mathbf{1}$ with $\sum_q A_{qi}\neq 0$ for all $i$ and draw $w_1$ from an absolutely continuous distribution; then simplicity holds almost surely.
    \item \textbf{Nonredundancy of overidentifying restrictions.} Let $R(\theta)=\vech_{\mathrm{off}}\!\big(\widetilde{\Lambda}(C_3(\theta))\,\Sigma(\theta)\,\widetilde{\Lambda}(C_3(\theta))^\top\big)$. Its Jacobian $J:=DR(\theta)\in\mathbb{R}^{\binom{d}{2}\times p}$ has full row rank $\binom{d}{2}$ at the population $\theta$.
\end{enumerate}
\end{assumption}

\noindent\textit{Comments.} Items (i) and (ii) coincide with the standing conditions used in Section~5 (see Lemma~\ref{lemma:g_diff} and Theorem~\ref{thm:main_sample_analogue_theorem}); item (iii) is the only \emph{new} condition here. It rules out local redundancy among the $\binom{d}{2}$ off‑diagonal restrictions so that the asymptotic covariance of $\sqrt{n}\,\hat r$ is nonsingular, yielding a $\chi^2_{\binom{d}{2}}$ limit for the Wald statistic.

Assumption~\ref{ass:appendix_joint}(i) delivers a multivariate CLT for $\hat\theta$ since $Z(X)$ contains monomials up to degree~3; Assumption~\ref{ass:appendix_joint}(ii) ensures differentiability of the eigenvector map as shown in Section~5; Assumption~\ref{ass:appendix_joint}(iii) is a standard local‑identification condition for the set of overidentifying restrictions (it holds generically).

\begin{lemma}[Differentiability of $R$]\label{lem:R_diff}
Under Assumption~\ref{ass:appendix_joint}(ii), the map
\[
R:\ \theta\mapsto\vech_{\mathrm{off}}\!\big(\,\widetilde{\Lambda}(C_3(\theta))\,\Sigma(\theta)\,\widetilde{\Lambda}(C_3(\theta))^\top\,\big)
\]
is continuously differentiable at $\theta$.
\end{lemma}

\begin{proof}
$C_3(\cdot)$ and $\Sigma(\cdot)$ are polynomial (hence analytic) in the coordinates of $\theta$. By Section~5 (Lemma~\ref{lemma:g_diff} and Theorem~\ref{thm:main_sample_analogue_theorem}), the oriented eigenvector map is analytic in a neighborhood of $C_3(\theta)$ under simple eigenvalues. Matrix products and $\vech_{\mathrm{off}}$ are smooth linear operations. Composition of smooth maps is smooth.
\end{proof}

\begin{theorem}[Asymptotic normality of $\hat r$]\label{thm:AR_normal}
Under Assumption~\ref{ass:appendix_joint}, we have
\[
\sqrt{n}\,\bigl(\hat r - r_0\bigr)
\;=\;
J\,\sqrt{n}\,(\hat\theta-\theta)\;+\;o_p(1)
\;\xrightarrow{d}\;
\mathcal{N}\!\bigl(0,\;\Omega\bigr),
\qquad
\Omega \;=\; J\,\Sigma_\theta\,J^\top,
\]
where $\Sigma_\theta:=\mathrm{Var}\!\big(Z(X)\big)$.
A consistent estimator is $\hat\Omega=\hat J\,\hat\Sigma_\theta\,\hat J^\top$, with
\[
\hat\Sigma_\theta=\frac{1}{n}\sum_{i=1}^n\bigl(Z(X_i)-\hat\theta\bigr)\bigl(Z(X_i)-\hat\theta\bigr)^\top,
\qquad
\hat J=DR(\hat\theta).
\]
\end{theorem}

\begin{proof}
By the multivariate CLT and Assumption~\ref{ass:appendix_joint}(i),
$\sqrt{n}(\hat\theta-\theta)\Rightarrow \mathcal{N}(0,\Sigma_\theta)$. By Lemma~\ref{lem:R_diff} and the Delta Method,
\[
\sqrt{n}\bigl(\hat r - r_0\bigr) = DR(\theta)\,\sqrt{n}(\hat\theta-\theta)+o_p(1) \Rightarrow \mathcal{N}(0,J\Sigma_\theta J^\top).
\]
Consistency of $\hat\Omega$ follows from the continuous mapping theorem and the consistency of the sample covariance, since $DR(\cdot)$ is continuous at $\theta$.
\end{proof}

\begin{theorem}[Wald $\chi^2$ test for uncorrelated structural errors]\label{thm:wald_chi2}
Under the null of uncorrelated structural errors, $r_0=0$. Under Assumption~\ref{ass:appendix_joint},
\[
T_n \;:=\; n\,\hat r^\top\,\hat\Omega^{-1}\,\hat r
\;\xrightarrow{d}\;
\chi^2_{\binom{d}{2}}.
\]
\end{theorem}

\begin{proof}
From Theorem~\ref{thm:AR_normal} with $r_0=0$, $\sqrt{n}\,\hat r \Rightarrow \mathcal{N}(0,\Omega)$ and $\hat\Omega\to_p \Omega$ with $\Omega$ nonsingular by Assumption~\ref{ass:appendix_joint}(iii). The claim follows by Slutsky’s theorem and the continuous mapping theorem for quadratic forms (Wald statistic).
\end{proof}

\section{Estimated vs.\ True Residuals}

This appendix shows that replacing the oracle residuals $e_t$ by estimated residuals
$\hat e_t$ has no first–order effect on our statistics. The proof follows two steps:
(i) a residual–replacement lemma for raw moments up to degree~3; and
(ii) a generic $\Delta$‑method invariance statement. A short corollary then covers
our over‑identification test.

\paragraph{Set‑up.}
Let a mean zero process $\{Y_t\}$ follow a stable VAR($p$):
\[
Y_t \;=\; A_1 Y_{t-1}+\cdots+A_p Y_{t-p}+ e_t,
\qquad
e_t \;=\; B\,u_t,
\qquad
\mathbb E[Y_t]=0
\]
where $u_t\in\mathbb R^{d}$ are i.i.d., $\mathbb E[u_t]=0$, and have finite moments
up to order $6+\eta$ for some $\eta>0$. The lag length $p$ is fixed and correctly specified.
Let $\hat A_\ell$ be the OLS estimators and define
\[
\hat e_t \;:=\; Y_t - \sum_{\ell=1}^p \hat A_\ell Y_{t-\ell},
\qquad
\delta_t \;:=\; \hat e_t - e_t \;=\; -\sum_{\ell=1}^p (\hat A_\ell-A_\ell)\,Y_{t-\ell}.
\]
Let $Z(\cdot)$ stack all raw monomials of total degree $1$–$3$ of its $d$‑vector
argument (as in Section~5). Write
\[
\bar Z_n \;:=\; \frac{1}{n}\sum_{t=1}^n Z(e_t),
\qquad
\hat Z_n \;:=\; \frac{1}{n}\sum_{t=1}^n Z(\hat e_t),
\qquad
\theta_0 \;:=\; \mathbb E[Z(e_t)] .
\]

\paragraph{Assumptions.}
\begin{enumerate}
\item \emph{Stability \& moments.} The VAR is stable; $\{Y_t\}$ is strictly stationary
and ergodic with $\mathbb E\|Y_t\|^{6}<\infty$.
\item \emph{First‑stage rate.} With fixed $p$,
$\sqrt{n}\,\mathrm{vec}(\hat A_1-\!A_1,\dots,\hat A_p-\!A_p)=O_p(1)$.
\item \emph{Smoothness at the population.} Any map $G$ we apply to the raw‑moment
vector (centering $\to$ cumulants $\to$ Hessians $\to$ inverse/product $\to$ eigenvectors
with a fixed orientation rule $\to$ demixed covariance $\to$ off‑diagonals) is continuously
differentiable at $\theta_0$. \emph{Comment:} In our application this holds because
the centering/cumulant maps are polynomial; inversion is smooth on nonsingular neighborhoods;
and the eigenvector map is analytic under simple eigenvalues, which we assume in Section~5.
\end{enumerate}

\paragraph{Norm convention.}
Vectors use $\|\cdot\|_2$; matrices use the operator norm $\|\cdot\|_{\mathrm{op}}$.
All $O_p(\cdot)$ orders are invariant to replacing $\|\cdot\|_{\mathrm{op}}$ by an equivalent norm
(e.g., Frobenius).

\begin{lemma}[Residual replacement is $o_p(n^{-1/2})$ in mean–zero VAR]
\label{lem:RR-littleo-mz}
Under the setup above,
\[
\sqrt{n}\,\bigl\|\hat Z_n-\bar Z_n\bigr\| \;\xrightarrow{p}\; 0 .
\]
\end{lemma}

\begin{proof}
Fix a coordinate $z(\cdot)$ of $Z(\cdot)$ (total degree $\le 3$). A second–order Taylor expansion at $e_t$
yields
\[
z(\hat e_t)-z(e_t)
= \underbrace{\nabla z(e_t)^\top \delta_t}_{L_t}
\;+\;
\underbrace{\tfrac12\,\delta_t^\top H_t\,\delta_t}_{Q_t},
\qquad
H_t:=\nabla^2 z(e_t+\tau_t\delta_t),\ \tau_t\in(0,1).
\]

\emph{Quadratic term.} Since $z$ is at most cubic 

\begin{align*}
|Q_t|
&= \tfrac12 \big|\delta_t^\top H_t \delta_t\big|
\;\le\; \tfrac12 \,\|H_t\|_{\mathrm{op}}\,\|\delta_t\|_2^2
\quad\text{(since }|x^\top M x|\le \|M\|_{\mathrm{op}}\|x\|_2^2\text{)}\\
\frac1n\sum_{t=1}^n |Q_t|
&\le\; \frac12\,\frac1n\sum_{t=1}^n
\underbrace{\|H_t\|_{\mathrm{op}}}_{=:a_t}\,
\underbrace{\|\delta_t\|_2^2}_{=:b_t}
\;\le\; \frac12
\Big(\frac1n\sum_{t=1}^n a_t^2\Big)^{\!1/2}
\Big(\frac1n\sum_{t=1}^n b_t^2\Big)^{\!1/2}\\
&=\; \frac12
\Big(\frac1n\sum_{t=1}^n \|H_t\|_{\mathrm{op}}^2\Big)^{\!1/2}
\Big(\frac1n\sum_{t=1}^n \|\delta_t\|_2^4\Big)^{\!1/2},
\end{align*}

\[
\frac1n\sum_{t=1}^n |Q_t|
\;\le\; \frac12
\Big(\frac1n\sum \|H_t\|_{\mathrm{op}}^2\Big)^{1/2}
\Big(\frac1n\sum \|\delta_t\|_2^4\Big)^{1/2}.
\]
Write $\delta_t=-\Delta W_t$, $W_t:=(Y_{t-1}^\top,\dots,Y_{t-p}^\top)^\top$.
Then $\|\delta_t\|\le \|\Delta\|_{\mathrm{op}}\|W_t\|$, with $\|\Delta\|_{\mathrm{op}}=O_p(n^{-1/2})$,
and stability gives $\frac1n\sum \|W_t\|^4=O_p(1)$ and $\frac1n\sum \|H_t\|_{\mathrm{op}}^2=O_p(1)$.
Hence $\frac1n\sum |Q_t| = O_p(n^{-1})$.

Note: 
\[
H_t \;=\; \nabla^2 z\big(e_t+\tau_t \delta_t\big),
\quad z \text{ has total degree }\le 3.
\]
For such $z$, the Hessian is affine: $\nabla^2 z(x)=H_0+\sum_{j=1}^d x_j H_j$, hence
\[
\big\|\nabla^2 z(x)\big\|_{\mathrm{op}} \;\le\; c_0 + c_1 \|x\|_2.
\]
Evaluating at $x=e_t+\tau_t\delta_t$,
\[
\|H_t\|_{\mathrm{op}}
\;\le\; c_0 + c_1\big(\|e_t\|_2+\|\delta_t\|_2\big)
\;=\; O_p(1)+o_p(1)\;=\;O_p(1),
\]

Thus $\|H_t\|_{\mathrm{op}}^2=O_p(1)$.
\[
\frac1n\sum_{t=1}^n \|H_t\|_{\mathrm{op}}^2
\;\le\; C\Big(1+\frac1n\sum_{t=1}^n \|e_t\|_2^2 + \frac1n\sum_{t=1}^n \|\delta_t\|_2^2\Big)
\;=\; O_p(1),
\]

\emph{Linear term.} Using $\delta_t=-\sum_{\ell} \Delta A_\ell Y_{t-\ell}$,
\[
\frac1n\sum_{t=1}^n L_t
= -\sum_{\ell=1}^p \Big\langle \Delta A_\ell,\ \frac1n\sum_{t=1}^n \nabla z(e_t)\,Y_{t-\ell}^\top \Big\rangle .
\]

note: \[
\langle U,V\rangle \;:=\; \operatorname{tr}(U^\top V)\;=\;\sum_{i,j} U_{ij}V_{ij},
\qquad
a^\top M b \;=\; \operatorname{tr}(M^\top a b^\top)\;=\;\langle M,\; a b^\top\rangle.
\]

Because the model is mean–zero, $\mathbb{E}[Y_{t-\ell}]=0$, and $e_t$ is independent of $Y_{t-\ell}$,
so $\mathbb{E}[\nabla z(e_t)\,Y_{t-\ell}^\top]=\mathbb{E}[\nabla z(e_t)]\,\mathbb{E}[Y_{t-\ell}]^\top=0$.
A CLT for these square‑integrable stationary terms gives
$\frac1n\sum \nabla z(e_t)Y_{t-\ell}^\top = O_p(n^{-1/2})$; multiplying by
$\|\Delta A_\ell\|=O_p(n^{-1/2})$ yields $\frac1n\sum L_t = O_p(n^{-1})$.

Combining and stacking the finitely many coordinates of $Z(\cdot)$:
$\|\hat Z_n-\bar Z_n\|=O_p(n^{-1})$, hence $\sqrt{n}\,\|\hat Z_n-\bar Z_n\|\to_p 0$.
\end{proof}

\begin{theorem}[Same first–order limit for $G(\hat Z_n)$ and $G(\bar Z_n)$]
\label{thm:same-first-order}
Let $G:\mathbb{R}^m\to\mathbb{R}^k$ be continuously differentiable at $\theta_0=\mathbb{E}[Z(e_t)]$.
Here $G$ can be the composition used in the paper: raw $\to$ centering $\to$ third cumulants
$\to$ Hessians $\to$ inverse/product $\to$ eigenvectors (simple spectrum with a fixed, continuous orientation rule)
$\to$ demixed covariance $\to$ unique off‑diagonals. If
$\sqrt{n}\,(\bar Z_n-\theta_0)\Rightarrow \mathcal N(0,\Sigma_Z)$, then by
Lemma~\ref{lem:RR-littleo-mz}
\[
\sqrt{n}\,\bigl(G(\hat Z_n)-G(\theta_0)\bigr)
\;-\;
\sqrt{n}\,\bigl(G(\bar Z_n)-G(\theta_0)\bigr)
\;\xrightarrow{p}\; 0,
\]
and therefore
\[
\sqrt{n}\,\bigl(G(\hat Z_n)-G(\theta_0)\bigr)
\;\Rightarrow\;
\mathcal N\!\bigl(0,\; D G(\theta_0)\,\Sigma_Z\,D G(\theta_0)^\top\bigr).
\]
\end{theorem}

\begin{corollary}[Wald test based on estimated residuals]
\label{cor:wald-mz}
Let $\hat r$ stack the strict upper‑triangular elements of the demixed covariance produced by $G(\hat Z_n)$.
Under the null of uncorrelated structural shocks and the identification/separation conditions in the paper
(simple eigenvalues; fixed, continuous orientation), we have
\[
\sqrt{n}\,\hat r \;\Rightarrow\; \mathcal N(0,\Omega),
\qquad
\Omega \;=\; D R(\theta_0)\,\Sigma_Z\,D R(\theta_0)^\top,
\]
and the Wald statistic $T_n := n\,\hat r^\top \hat\Omega^{-1}\hat r \Rightarrow \chi^2_{\binom{d}{2}}$
for any consistent $\hat\Omega$.
\end{corollary}

\newpage
\bibliographystyle{apalike}
\bibliography{references}  

\end{document}